\documentclass[12pt]{article}

\usepackage[margin=1in]{geometry}
\usepackage{setspace}
\usepackage{amsmath,amssymb,amsthm,mathtools}
\usepackage[authoryear,round]{natbib}
\def\endbibitem{}
\usepackage{booktabs}
\usepackage{array}
\usepackage{threeparttable}
\usepackage{graphicx}
\usepackage{float}
\usepackage[section]{placeins}
\usepackage{enumitem}
\usepackage{microtype}
\usepackage{xcolor}
\newenvironment{fullwidthtablenotes}{%
  \par\smallskip
  \noindent\begin{minipage}{\linewidth}%
  \footnotesize\normalfont
  \setlength{\parindent}{0pt}%
  \setlength{\parskip}{0pt}%
  \setlength{\leftskip}{0pt}%
  \setlength{\rightskip}{0pt}%
  \setlength{\parfillskip}{0pt plus 1fil}%
  \textit{Notes:}\enspace\ignorespaces
}{%
  \end{minipage}\par
}

\usepackage[colorlinks,citecolor=blue,linkcolor=blue,urlcolor=blue,
  hypertexnames=false,
  pdftitle={The Race for Elite Destinations: Education Competition and Low Fertility in Korea},
  pdfauthor={Dongwoo Kim}]{hyperref}

\DeclareMathOperator{\Cov}{Cov}
\DeclareMathOperator{\Var}{Var}
\newcommand{\E}{\mathbb{E}}

\newcommand{\1}{\mathbf{1}}
\newcommand{\dd}{\,\mathrm{d}}

\allowdisplaybreaks

\theoremstyle{plain}
\newtheorem{proposition}{Proposition}
\newtheorem{lemma}{Lemma}
\newtheorem{corollary}{Corollary}
\theoremstyle{definition}
\newtheorem{definition}{Definition}
\newtheorem{remark}{Remark}

\renewcommand{\thesection}{\Roman{section}}
\renewcommand{\thesubsection}{\thesection.\Alph{subsection}}
\renewcommand{\thesubsubsection}{\thesubsection.\arabic{subsubsection}}
\newcommand{\vEffect}{0.240}
\newcommand{\vContest}{1.8463}
\newcommand{\vLottery}{2.0862}

\newcommand{\vRaceMoneyPct}{5.0}

\newcommand{\vPointRmse}{0.64}
\newcommand{\vBE}{5.35}

\newcommand{\vPosSharePct}{77.3}

\newcommand{\vFrontier}{2.082}
\newcommand{\vFrontierShare}{98}
\newcommand{\vLotteryProdCostBenchmark}{0.7}
\newcommand{\vLotteryProdCost}{1.6}

\newcommand{\vTransferFive}{0.031}
\newcommand{\vTransferTwo}{0.012}

\newcommand{\vEffectYoung}{0.275}
\newcommand{\vKTYrmseYoung}{0.21}
\newcommand{\vCohortRepro}{76}

\newcommand{\vCEexTime}{16.3}
\newcommand{\vCEexTimeNoShock}{45.7}
\newcommand{\vQualifiedFert}{0.181}
\newcommand{\vCEqualified}{8.8}
\newcommand{\vQualifiedOutput}{0.988}
\newcommand{\vCEtransferFive}{2.5}
\newcommand{\vTaxTransferFert}{0.003}
\newcommand{\vCETaxTransfer}{0.9}
\newcommand{\vTaxTransferRate}{0.76}
\newcommand{\vCEcapEtaOne}{0.4}
\newcommand{\vCEcapLot}{18.0}
\newcommand{\vCEgenOne}{17.2}

\newcommand{\vCEabsSlots}{-19.8}
\newcommand{\vCEphiOne}{-0.1}
\newcommand{\vPhiTotal}{3.9}
\newcommand{\vCEctcHi}{20.2}

\newcommand{\vChildTimeRelease}{1.66}
\newcommand{\vChildTimeShare}{92}
\newcommand{\vDWfull}{1.81}

\title{The Race for Elite Destinations:\\
Education Competition and Low Fertility in Korea\thanks{Kim: Simon Fraser University and Korea University (dongwook@sfu.ca). I gratefully acknowledge support from the Social Sciences and Humanities Research Council of Canada under Insight Grant 435-2024-0322. All errors are my own.}}
\author{Dongwoo Kim}
\date{August 2026}

\begin{document}

\maketitle

\begin{spacing}{1.10}
\begin{abstract}
South Korea has the world's lowest fertility and an intense education race.
Families devote nine percent of lifetime income to education, mostly to private tutoring
with near-zero measured returns. I show that competition for coveted careers
generates an assignment externality: when all families spend more,
the admission bar rises and children become costly. In a quantitative model
calibrated to this mechanism, replacing score-based assignment with a
capacity-preserving lottery raises completed fertility by 0.24 children per
couple. Pronatal transfers deliver one eighth of that increase. Education taxes
deliver almost none. Preference shifts, not a fiercer race, explain the
decline across cohorts.
\end{abstract}
\end{spacing}

\noindent\textit{Keywords:} low fertility; education competition; assignment externality; South Korea.\\
\textit{JEL codes:} J13, I24, I28, D62.

\thispagestyle{empty}
\newpage

\begin{quote}
\small\itshape
``[E]very Korean child came home from grade school, and worked with a tutor
for four full hours in the afternoon and the evening, driven by these
Tiger Moms. Are you surprised when you lose to people like that? Only
if you're a total idiot.''

\upshape\raggedleft -- Charlie Munger, quoted in David Clark, \emph{Tao of Charlie Munger} (2017)
\end{quote}

\section{Introduction}
\label{sec:intro}

South Korea has the world's lowest fertility and an exceptionally intense
education race. Private academies, tutoring, and test preparation figure prominently in family budgets and public
debate. The race is widespread across the income distribution and culminates in competition for admission to selective universities and entry into preferred careers. These pressures have long been described as Korea's ``education fever''
\citep{Seth2002,Sorensen1994,AndersonKohler2013}. The same phenomenon appears across
East Asia, where educational competition coexists with very low
fertility. \citet{FangLiu2026} call this the \textit{Confucian fertility
paradox}: the education tradition ``turns a once pro-natal culture into a powerful
engine of low fertility.'' This paper asks whether the education race itself
raises the cost of children and reduces family size.

Educating a child in Korea absorbs roughly nine percent of cumulative after-tax
family income over the first twenty-five years. Private tutoring accounts for
nearly three quarters of that spending.\footnote{In the Korean Labor and Income
Panel Study data, 77 percent of households purchase private education for their
children in a given year.} Yet the payoff to private-education spending is hard to detect. In
the Korean Education and Employment Panel (KEEP) data, private-education spending
predicts no measurable improvement in college-entrance scores, entry into a
preferred career, or wages at age~29. The estimated returns are too small to
explain spending on this scale. \citet{KTY2024}, hereafter KTY, attribute the
heavy spending to status externalities: parents value
their children's human capital relative to that of other children, so one
family's investment induces others to spend more. I study a different source of
strategic interaction, an \textit{assignment externality}. Education helps determine who
gains access to scarce and highly valued opportunities, from selective
universities to licensed professions and secure careers in large firms and the
public sector. Because access to such positions is limited, when all
families spend more on their children's education, the bar that each child
must clear rises. Shared aspirations for
these opportunities can therefore generate an education race even when families
have no direct taste for rank.

The paper's main contribution is to embed this assignment externality in a quantitative model of fertility, supported by new evidence on the returns to private education, including within-school estimates linking spending to college-entrance scores and early-career wages. Combining these estimates with household spending data separates productive investment from positional spending and measures how the race raises the equilibrium cost of a child. By locating the source of competition in the interaction between shared aspirations and limited capacity, the model expands the policy set beyond cash transfers and corrective taxes to reforms that change the supply of preferred positions or how they are allocated.

I first formalize the education race using an analytical model of an all-pay contest for a
fixed-capacity destination. Each family treats the admission cutoff as fixed and therefore has a private incentive to raise its child's score. When all families spend more on positional effort, the cutoff rises and the allocation does not change. Positional spending is privately valuable but largely dissipated in aggregate, even though families have no utility from rank, relative income, or relative human capital. Competition also raises the cost of family size:
each additional child brings another education bill and requires the family to
enter the contest again. The analytical model does not, however, determine which
policy works best. Changing the admission rule can reduce the return to
positional effort, whereas adding more positions at preferred destinations can either
weaken or intensify the race, depending on how the value of each position changes
as access expands. The effects of education taxes depend on how much spending is
productive. Ranking these policies requires measuring the productive and
positional shares of observed spending and how households respond in equilibrium
when the rules of the contest change.

I therefore embed the contest in a quantitative overlapping-generations model.
Households choose fertility and labor supply together with productive education,
positional effort, and whether to enter the contest. Productive education raises human capital regardless of
destination, whereas positional effort changes admission probabilities without
raising human capital. I calibrate the model to married or cohabiting couples in
the Korean Labor and Income Panel Study (KLIPS) whose female partner was born
between 1970 and 1975. The KLIPS data pin down how much families spend on private
education and how its share of income varies across households. The near-zero
estimated returns to private education in KEEP cap how much of this spending the
model can attribute to productive investment. Together, these moments imply that
most observed private-education spending is positional. Contest outlays absorb
\vRaceMoneyPct\ percent of household income.

The quantitative model shares its household, human-capital, and fertility blocks
with KTY, who show that status externalities can generate excessive education
investment and low fertility. KTY also study a college-admission tournament and
report that total college attendance remains roughly constant across their policy
experiments. They point to capacity and admission rules as relevant policy margins.
\citet{MahlerTertiltYum2025} note that comparison motives can themselves
be read as a reduced-form representation of competition created by scarce
high-quality college places. This paper makes that institutional structure
explicit. In KTY's benchmark, the externality originates in a preference for
relative human capital. Because preferences are hard to change, policy in that
framework works through taxes and transfers that dampen its consequences. Here
families have no relative rank concern. The externality arises when widely shared
aspirations meet a fixed-capacity, score-based allocation rule. Because its
source is institutional, the model can compare assignment
reform with capacity expansion, qualified lotteries, and education taxes.
KTY's tournament extension keeps a single education input
that both builds human capital and improves admission. Here the two inputs are
separate choices and entry into the contest is endogenous.

These differences matter for policy. In the main counterfactual experiment, I
replace score-based assignment with a lottery that leaves the admission share
and value of preferred destinations unchanged.
Productive education remains valuable, but positional spending no longer
affects who receives those destinations, and completed fertility rises by
\vEffect\ children per couple. Random assignment discards some productive
sorting, but output per unit of labor falls by only
\vLotteryProdCostBenchmark\ percent relative to the current system. The loss is
small because the positional outlay builds no human capital and productive
education rises slightly after those resources are released. Transfers and
education taxes deliver much smaller gains. A balanced-budget per-child
transfer equal to five percent of the average wage raises fertility by only
\vTransferFive\ children per couple, about one eighth of the lottery effect.
A tax on total private-education spending, and even an infeasible oracle tax
that targets positional outlays alone, have almost no effect once their revenue
is rebated. Adding positions at preferred
destinations has no uniform effect: broader access can attract new contestants,
while the value of each position may fall.

The lottery also raises welfare after accounting for this productivity cost.
With the admission share held fixed, the partial-accounting measure that excludes
the unpriced child-time term delivers an average gain equivalent to \vCEexTime\
percent of consumption. The gains are largest among lower-income households.
Households in the top quintile lose because their resources bought a higher
admission probability under score-based assignment.

Education competition can keep fertility low without driving its continued
decline. The fall in completed fertility across cohorts raises a second
question: did the education race intensify across generations, or did desired
family size change? I examine this question using completed fertility for women
born in 1970--75 and 1976--81. Holding the contest and
education technology at their older-cohort values, I recalibrate preferences
and heterogeneity for the younger cohort. This specification reproduces
\vCohortRepro\ percent of the decline in completed fertility between the two
cohorts. The younger cohort draws less utility from children at every family size. A specification that instead allows only the contest and
education technology to change has an RMSE more than forty times as large and implies a
weaker, not stronger, contest. Competition nevertheless remains quantitatively
important: the same lottery would raise completed fertility in the younger
cohort by \vEffectYoung. The comparison points to a generational shift in
family-size preferences: the education race depresses fertility in both
cohorts, while preference changes account for most of the decline between them.

The mechanism connects the quantity--quality tradition
\citep{Becker1960,BeckerLewis1973,BeckerTomes1976} to models of costly competition
for fixed prizes
\citep{Akerlof1976,Lazear1981,FrankCook1995,BodohCreedHickman2018}. Standard
fertility models take the price of child quality as given. Here the cost of
improving a child's chance of reaching a preferred destination rises with other
families' spending, and the number of potential contestants is itself a fertility
choice. This interaction links the escalation in parental investment documented
by \citet{RameyRamey2010} and the intensive-parenting choices studied by
\citet{DoepkeZilibotti2017} to fertility, extending quantitative models in which
households solve their education problem in isolation
\citep{deLaCroixDoepke2003,JonesEtAl2010,DoepkeEtAl2023}. Two recent papers
study the admissions race directly. \citet{MahlerTertiltYum2025} use a
one-child college-admission model to show how fixed college positions can make application
preparation socially wasteful and raise the cost of children, while
\citet{Kang2026} estimates a dynamic Korean admissions tournament and quantifies
competition-driven tutoring. The first holds family size fixed in its admissions
example, and the second does not study fertility. This paper embeds the
admissions tournament in a fertility model and measures how its resource cost
changes completed fertility and the ranking of policy reforms.
In the Korean context,
the paper brings together evidence on education competition, family institutions,
and fertility \citep{AndersonKohler2013,HongEtAl2016,MyongParkYi2021} and gives
quantitative content to the conjecture in \citet{FangLiu2026} that reforming
educational pathways can achieve what cash transfers cannot.

This paper is organized as follows. Section~\ref{sec:empirical} presents the
empirical evidence. Section~\ref{sec:analytical} develops the analytical contest
model. Section~\ref{sec:model} presents the quantitative model.
Section~\ref{sec:calibration} calibrates the model for the benchmark cohort
(1970--75), then reports and decomposes the main counterfactual. Section~\ref{sec:policy} compares policy
experiments and evaluates welfare. Section~\ref{sec:cohorts} studies cohort
change, and Section~\ref{sec:conclusion} concludes. The Supplemental Appendix
contains the data construction, proofs, model and computational details, robustness
results, and supplementary welfare analysis.

\section{Fertility, Education, and Preferred Destinations}
\label{sec:empirical}

This section assembles the empirical foundation of the model: what raising a
child costs, what families compete for, and what the competitive spending
measurably returns. Section~\ref{sec:fertility_education} documents that
fertility is low, falling across cohorts, and mildly increasing in income,
while education investment is nearly universal, proportionally heaviest for
low-income families, and paid in the child's time as well as money.
Section~\ref{sec:preferred_destinations} identifies the prize: a narrow set of
preferred career destinations whose value current wages understate.
Section~\ref{sec:keep_returns} estimates the return to that spending and finds
no measurable return in examination scores, destination entry, or wages after
conditioning on income and achievement.

The household evidence comes primarily from the Korean Labor and Income Panel
Study (KLIPS), which I use through wave 27 (2024). Following KTY, the
benchmark cohort consists of married or cohabiting couples in which the woman
was born between 1970 and 1975 and appears at ages 40--43 in at least three
waves; the younger cohort applies the same construction to women born between
1976 and 1981. The 2024 KLIPS cross-section measures preferred-destination
population shares among adults ages 30--45 and wage premiums among wage workers
ages 25--54. The Korean
Education and Employment Panel (KEEP) provides a separate cohort for
estimating the return to private education. The Private Education Expenditures
Survey and the Korean Time Use Survey provide national spending, participation, and
study-time evidence; official admissions, workplace-preference, and youth
labor-force statistics provide the remaining institutional evidence.
Supplemental Appendix~A documents the sample rules and variable construction,
reproduces KTY's published moments under their original conventions, and reports
the KEEP sample-selection checks.

\subsection{Fertility and Education Spending}
\label{sec:fertility_education}

Table~\ref{tab:cohortdata} reports the harmonized cohort moments. Completed
fertility counts children observed by age 43, and permanent income averages
CPI-adjusted household income over ages 40--43. The education moments follow
the child's life cycle, using every observation in which a woman from
the relevant birth cohort has a resident child aged 0--24.\footnote{Restricting mothers to ages
40--43 would select only one part of the child's education life cycle.} The
table notes and the Supplemental Appendix give the remaining definitions.

Completed fertility falls from 1.873 in the 1970--75 cohort to 1.786 in the
1976--81 cohort. The decline is concentrated at two margins: the two-child
share falls by 5.7 percentage points and childlessness rises from 3.3 to 6.9
percent, while the one-child share increases more modestly and the
three-or-more share is stable. Fertility increases mildly with permanent income. In the benchmark
cohort it rises from 1.77 in the bottom quintile to 1.99 in the top, an elasticity
of 0.079, while childlessness falls from 6.5 to 1.3 percent; the younger
cohort's gradient is flatter, rising from 1.65 to 1.83 with an elasticity of
0.068. Both cohort gradients use completed fertility top-coded at three. The positive
gradient suggests that child costs weigh more heavily on lower-income households.\footnote{Models in which child costs are dominated by the time cost of parenting predict a negative income gradient of fertility \citep{deLaCroixDoepke2003}, while \cite{BaudinDeLaCroixGobbi2015} show that a goods-cost channel (a minimum consumption requirement for childbearing) works in the opposite direction at the extensive margin. In Korea, the gradient is positive at both margins, suggesting that the goods-cost channel dominates, consistent with large quasi-fixed expenditures on private education.}

\begin{table}[htbp]
\centering
\small
\caption{Harmonized cohort moments}
\label{tab:cohortdata}
\begin{threeparttable}
\begin{tabular}{lcc}
\toprule
 & 1970--75 & 1976--81 \\
\midrule
Households & 770 & 1,129 \\
Childless share & 0.033 & 0.069 \\
One-child share & 0.203 & 0.220 \\
Two-child share & 0.625 & 0.568 \\
Three-or-more share & 0.140 & 0.143 \\
Completed fertility & 1.873 & 1.786 \\
Work hours & 0.292 & 0.278 \\
Total education spending/income & 0.088 & 0.082 \\
Fertility--income elasticity & 0.079 & 0.068 \\
Total education spending--income elasticity & 0.659 & 0.627 \\
Bottom-quintile childlessness & 0.065 & 0.097 \\
Income Gini & 0.263 & 0.241 \\
\bottomrule
\end{tabular}
\end{threeparttable}
\begin{fullwidthtablenotes}
KLIPS through wave 27 (2024). Fertility, hours, and income moments use couples with the wife aged
40--43 and at least three observations. The two education moments use the whole
birth cohort with children aged 0--24. Work hours are
spouse-average weekly hours divided by 100. The younger cohort's education
moments are right-truncated: women born in 1976--81 who are observed with
secondary-school-age children gave birth unusually early.
\end{fullwidthtablenotes}
\end{table}

Education investment is nearly universal, and it is not receding. The
national evidence comes from the Private Education Expenditures
Survey, conducted jointly by Statistics Korea and the Ministry of Education,
which has measured per-student participation annually since 2007. In its 2019
round, three quarters of students received some form of private education, and
participation in academic subjects remained near 58 percent through high
school, as KTY report in their Table 1. At the high-school level, some participation in arts
and athletics also reflects preparation for specialized admissions, so the
reach of admission-oriented private education is broader than academic-subject
participation alone. Participation fell from 77.0 percent in 2007 to 67.8 percent
in 2016, dropped to 66.5 percent during the pandemic in 2020, recovered within
a year, and rose to a record 80.0 percent in 2024
\citep{StatisticsKoreaPEES2025}. Whatever drove the
decline in the early 2010s, the latest observed participation rate is the
highest in the series.

The KLIPS sample shows the same reach at the household level.
Among benchmark-cohort households with a resident child aged 0--24, 76.6
percent report positive private-education spending, rising from 57.2 percent in the
bottom permanent-income quintile to 85.0 percent in the top.

The burden is large over a child's life, and it falls with income. Lifetime
education spending equals 8.8 percent of after-tax household income, with
spending and income summed from birth through age 24. Following
KTY, I measure this in one-child households, where the expenditure
is unambiguously attributable to a single child. Private education is 72.8 percent
of that total, 6.4 percent of income; the remainder is school fees and other
public-schooling outlays. Because publicly provided education is common across
households, the model targets the household-specific private component. The
elasticity of total education spending with respect to income is 0.659, well below one. Across household-years in the education sample,
the annualized spending-to-income ratio averages 20.9 percent in the bottom
permanent-income quintile and 12.8 percent in the top. Lower-income
families participate widely and devote the largest share of income to education,
a pattern that can reinforce the positive fertility--income gradient.

The race also consumes the child's own time. In the 2024 Korean Time Use Survey,
children aged 10--18 spend 18.2 hours per week in out-of-school study (private-academy
instruction accounts for 9.1) and a further 6.4 hours traveling to and from study,
on top of 25.5 hours of school instruction. These hours rise strongly with
household income, from 13.7 hours per week in low-income households to 25.9 in the
top band, and rose about 17 percent between 2004 and 2024 while fertility fell.
Per-child study time is nearly flat in the number of children in the household
(18.5, 18.6, and 17.4 hours for one, two, and three or more minors), so in the time
dimension the household burden scales with family size rather than being divided
among siblings. The Supplemental
Appendix records the construction of these statistics.

\subsection{Preferred Destinations}
\label{sec:preferred_destinations}

Preferred destinations are the jobs toward which Korean education and employment
preparation are commonly directed. The composite category
contains five mutually exclusive components: licensed professions, covering
selected medical, legal, and accounting occupations; university and
school teachers; regular government employment; regular employment in public
institutions and public enterprises; and regular employment in large private or
foreign firms. Academy instructors and irregular public-sector workers are excluded.
These jobs provide the attributes Korean society
prizes in a career: high pay in the professions and large firms, exceptional
job security in the public sector and teaching, and social standing
throughout. They are also scarce. Entry is controlled by competitive gates:
national licensing examinations for the professions, certification examinations
for teachers, entry examinations for government and public institutions, and
the graduate recruitment programs of large firms. The
preference has deep roots: government office carried social prestige through
centuries of dynastic civil-service examinations, and the modern examination
routes inherit that standing \citep{Sorensen1994,Seth2002}.
Selective universities are intermediate destinations on the path to these careers.
I therefore measure preferred destinations at the career stage and, in the
model, compress the intervening admissions, examinations, and recruitment
gates into a single score-based contest.

The demand response to the 2024 expansion of medical-school places shows how
widely the aspiration for these destinations is shared. The nationwide medical-school intake
had been fixed at 3,058 students per year for almost two decades. When the
government announced 2,000 additional annual places in February 2024,
high-school graduate applicants for the College Scholastic Ability Test (CSAT),
the national university entrance examination, rose to 161,784, or 31.0 percent
of all registrants, the highest number since the 2004 academic year
\citep{KICE2024}. At the same time, 89 private academies marketed 136
medical-track programs for elementary-school students \citep{Asiae2024}, and
private-education spending reached a record 29.2 trillion won despite a shrinking
student population. Spending grew fastest in the nonmetropolitan regions whose
residents qualify for regional admission quotas at local medical schools
\citep{StatisticsKoreaPEES2025}. The government
restored the intake to 3,058 for 2026 \citep{MOE2025}. One gate widened, and
demand across the country pivoted to it within months.

Table~\ref{tab:destination} reports each component's population share and
conditional wage premium; Supplemental Appendix~A gives the underlying
occupation and workplace codes. In the KLIPS 2024 cross-section,
the composite destination accounts for 18.45 percent of adults ages 30--45. Regular large or foreign firms are
the largest component at 8.77 percent, followed by government and public
institutions at roughly 3 percent each, teachers at 2.42 percent, and licensed
professions at 0.94 percent. The conditional current-wage premium varies sharply across components. Licensed
professions and large firms have premiums of 0.428 and 0.317 log points; teachers,
government, and public institutions sit between 0.036 and 0.103. A model that
valued these destinations only through current wages would therefore predict much
weaker demand for the public and teaching components than their observed
popularity suggests.\footnote{The estimates are close to independent evidence:
using graduate employment records linked to business-group registries,
\citet{KwonLee2019} estimate premiums of 0.228--0.264 for large business groups
among college-graduate youth, and the composite premium of 0.257 lies within
that range.} Demand for these destinations, including the low-premium components, is widespread: in the
2025 Social Survey, 63.1 percent of young respondents preferred large firms,
public enterprises, or government institutions as workplaces
\citep{StatisticsKoreaSocial2025}, and among the 565,000 nonemployed young people
preparing for employment examinations in 2024, 23.2 percent were preparing for the
civil service \citep{StatisticsKoreaYouth2024}. Current wages are therefore an incomplete
measure of what households are pursuing.

\begin{table}[htbp]
\centering
\small
\caption{The composite preferred destination}
\label{tab:destination}
\begin{threeparttable}
\begin{tabular}{lccc}
\toprule
Component & Population share (\%) & Log-wage premium & Sample size \\
\midrule
Licensed professions & 0.94 & 0.428 & 48 \\
University and school teachers & 2.42 & 0.103 & 207 \\
Regular government employment & 3.28 & 0.036 & 248 \\
Regular public institutions & 3.04 & 0.063 & 223 \\
Regular large or foreign firms & 8.77 & 0.317 & 733 \\
\midrule
Composite destination & \textbf{18.45} & \textbf{0.257} & 1,459 \\
\bottomrule
\end{tabular}
\end{threeparttable}
\begin{fullwidthtablenotes}
All statistics are measured in the KLIPS wave 27 (2024) cross-section.
Population shares use adults ages 30--45 and cross-sectional weights.
Wage-premium regressions use wage workers ages 25--54 and control for age,
sex, and education. The composite premium is 0.257 with a
standard error of 0.014 and a 95 percent confidence interval of
$[0.230,0.285]$. Professional self-employment is included.
\end{fullwidthtablenotes}
\end{table}

The missing value is concrete: job security, pensions and other deferred
compensation, predictable career progression, working conditions, occupational
licensing, and the option value of future careers. Some of these attributes are
pecuniary over the life cycle even though they do not appear in current wages.

\subsection{The Measured Return to Private Education}
\label{sec:keep_returns}

Finally, I investigate whether private-education spending has a measurable return in the outcomes KEEP observes, once parental
income and achievement are held fixed. I use KEEP to estimate the return to private education at three stages: the college-entrance exam score, entry into a
preferred destination, and wages at age~29. KEEP follows a cohort of students
who were high-school seniors in 2004. The household survey records private-education spending, income, and a
teacher assessment of achievement in 2004. Linked files provide the 2005
college-entrance examination score and employment outcomes in 2015.

\begin{table}[htbp]
\centering
\footnotesize
\setlength{\tabcolsep}{6pt}
\caption{Returns to private-education spending, KEEP 2004 cohort}
\label{tab:keep_returns}
\begin{threeparttable}
\begin{tabular}{lccc}
\toprule
& \shortstack{CSAT\\percentile} & \shortstack{Composite\\destination (logit)} &
\shortstack{Log wage\\at age 29} \\
\midrule
No controls & 2.071 & 0.090 & 0.004 \\
            & (0.344) & (0.041) & (0.008) \\
+ log household income & 1.107 & 0.049 & $-0.005$ \\
            & (0.400) & (0.048) & (0.009) \\
+ teacher achievement rating & 1.539 & 0.027 & 0.004 \\
            & (0.384) & (0.059) & (0.011) \\
+ region fixed effects & 1.252 & 0.018 & 0.003 \\
            & (0.380) & (0.064) & (0.011) \\
+ school fixed effects & \textbf{0.285} & -- & \textbf{0.002} \\
            & (0.323) & & (0.013) \\
\midrule
Observations (final specification) & 1{,}220 & 670 & 574 \\
\bottomrule
\end{tabular}
\end{threeparttable}
\begin{fullwidthtablenotes}
Coefficients on log monthly household private-education
spending. Robust standard errors in parentheses. The destination column is a logit for
employment at a firm with 300 or more employees or in government or a public
institution, observed in 2015. School fixed effects are not reported for the
destination column because the within-school variation in that binary outcome is too
thin at this sample size.
\end{fullwidthtablenotes}
\end{table}

Without controls, spending strongly predicts the CSAT percentile. Household
income cuts the coefficient roughly in half; adding the teacher achievement
rating partly restores it, to 1.539. School fixed effects reduce it further to
0.285, less than one quarter of
the region fixed-effects estimate. With a standard error of 0.323, the estimate
is not statistically distinguishable from zero. At the point estimate, doubling
monthly spending from the sample median of 350{,}000 won corresponds to a 0.20
percentile-point difference. Even the upper end of the ninety-five percent
interval corresponds to only 0.64 percentile points.
School fixed effects provide the relevant comparison because high-school
admission remained selective for about a third of this cohort: 35 of its 100
schools, enrolling 35.6 percent of the sample, admitted students by
examination when the cohort entered high school in 2002. A within-school teacher rank does not put students from selective
and non-selective schools on a common achievement scale. Without school effects,
sorting of high-spending families into high-scoring schools remains in the
spending coefficient.

The within-school variation remains substantial. The standard
deviation of monthly spending is 460{,}000 won within schools and 261{,}000 won
across school means. Only 16 percent of the variance in log spending lies between
schools. The teacher rating is itself a within-school rank, whereas 25 percent of
CSAT variation lies between schools. School effects absorb that outcome component
while preserving 92 percent of the variation in spending. The coefficient falls
more than fourfold and its standard error declines from 0.380 to 0.323. The split by
admission regime confirms that sorting matters. Estimated separately, the raw spending--CSAT
coefficient is $0.013$ (s.e.\ 0.476) in equalization regions, where students
are assigned to high schools by residential lottery, and $2.162$ (s.e.\ 0.624)
in selective-admission regions; adding school effects cuts the
selective-region estimate to $0.621$ (s.e.\ 0.457). The pattern points to
sorting into higher-scoring schools rather than a return to spending.

The destination and wage estimates are already close to zero after controlling
for income and achievement. The logit slope of preferred-destination entry with respect to log spending is
$0.027$ (s.e.\ 0.059). The wage elasticity is $0.004$ (s.e.\ 0.011) and becomes
$0.002$ (s.e.\ 0.013) with school effects. Its ninety-five percent interval is
$[-0.023,\,0.026]$, implying an earnings difference of less than two percent in
either direction for a doubling of spending.

The Korean literature also finds modest effects. \citet{Kang2012}
instruments subject-level tutoring expenditure among middle-school students and
estimates that a ten-percent increase in spending raises the subject score by at most
0.75--1.28 percent, depending on the subject. \citet{KuLeeKim2024} report positive effects on
college-entrance scores and interpret them as upper bounds because their
cross-subject instrument may retain common ability and parental support. Their
design controls for school type and residential location. My comparable region fixed-effects
estimate implies 0.87 percentile points for a doubling of spending; adding school
fixed effects lowers it to 0.20. Evidence on adult outcomes points in the same direction. In an independent
analysis of the same KEEP cohort, \citet{KimSeik2019} finds no significant
wage association with private tutoring after conditioning on college-entrance
performance. \citet{deSilva2021} uses staggered regional curfews on late-night tutoring as
instruments for private-tutoring spending. The average effect on college entry is close to zero.
Positive effects are concentrated among students whose
parents have less education, and returns decline where spending is already high.

The evidence calls for a quantitative model
with productive and positional inputs, a capacity-constrained destination,
value beyond the contemporaneous wage, and noisy transmission of family
advantage. Before introducing that structure, the next section isolates the
assignment externality generated by a score-responsive cutoff.

\section{The Assignment Externality}
\label{sec:analytical}

\label{sec:analytical_fertility}
This section builds the smallest environment that contains the assignment
externality: parents choose fertility, productive education, and positional
effort for each child, and a fixed share of children receives the preferred
destination by score ranking. Each family takes the admission cutoff as given,
so effort that raises the score carries a positive private return. Yet capacity
clearing moves the cutoff one-for-one with the common component of that effort,
leaving assignments unchanged. The effort
is privately rational and collectively dissipated, although nothing in
preferences rewards rank or relative standing.
A closely related fixed-capacity example appears in
\citet{MahlerTertiltYum2025}, where socially wasteful application preparation
raises the signal used to allocate a fixed number of college places. Their
example holds family size fixed. Here fertility and productive education are
choices, and the quantitative model adds noisy assignment and endogenous contest
entry. These margins make the cost of the admissions race part of the decision
to have another child.

A parent has resources $y$ and chooses the number of children $n$. The parent
chooses the same productive education $x\geq0$ and positional effort $e\geq0$
for each child, with a strictly convex per-child cost of effort.\footnote{Convexity is the
canonical tournament assumption \citep{Lazear1981}, and it is what sustains an
interior equilibrium: with a linear cost, a prize of this size would push
effort to corners rather than to the interior spending distribution documented
in Section~\ref{sec:empirical}. The measured score technology points the same
way: doubling spending moves the examination score by 0.2 percentile points,
so buying additional score is already expensive at observed levels.} Consumption is
\begin{equation}
c=y-nx-nC(e),
\qquad C(0)=C'(0)=0,\quad C'(e)>0,\quad C''(e)>0
\quad\text{for }e>0.
\label{eq:simple_budget}
\end{equation}
Productive education raises human capital through $h(x)$, with $h'>0$ and
$h''<0$; positional effort does not enter human capital. A child's contest score
is $z=a+\theta\log h(x)+g(e)+\varepsilon$, where $a$ is an exogenous family
endowment, $\theta>0$, $g'>0$, and $\varepsilon$ has a continuously differentiable
distribution $F_\varepsilon$ and density $f_\varepsilon$. A child receives a
preferred destination if $z\geq v$, so the admission probability is
\begin{equation}
p=1-F_\varepsilon\!\left(v-a-\theta\log h(x)-g(e)\right).
\label{eq:simple_probability}
\end{equation}
The destination raises log income by $\delta$ and carries an additional absolute
value $B$; write $\Delta\equiv\delta+B>0$ for the total value to the parent.
Conditional on parity (the number of children), the parent maximizes
\begin{equation}
u(c)+\Phi(n)+\beta n\bigl\{u_c(h(x))+\Delta\,p\bigr\},
\label{eq:simple_objective}
\end{equation}
where $u$ is increasing and strictly concave, $\Phi(n)$ collects the direct
value of parity, $u_c$ is the parent's valuation of the child's human capital
with $u_c'>0$ and $u_c''\leq0$, and $\beta>0$ weights child outcomes. The additive
form of the prize is exact when that valuation is logarithmic: a proportional
income premium then adds the same increment $\delta$ for every family, so it
sums with $B$. Nothing in preferences
depends on other children's outcomes. The cutoff $v$
clears a fixed capacity: the share of children admitted equals $q\in(0,1)$.

Taking $v$ as given, an interior effort choice satisfies
\begin{equation}
u'(c)\,C'(e)
=\beta\Delta\,
f_\varepsilon\!\left(v-a-\theta\log h(x)-g(e)\right)g'(e).
\label{eq:effort_foc}
\end{equation}
The number of children cancels because both the benefit and the cost are per
child. Suppose families share the same state and make the same choices. Capacity
clearing then requires
\begin{equation}
v-a-\theta\log h(x)-g(e)=F_\varepsilon^{-1}(1-q).
\label{eq:symmetric_cutoff}
\end{equation}

\begin{lemma}[Cutoff absorption]
\label{lem:absorption}
In a symmetric capacity-constrained equilibrium, any common increase in the
positional-score component $g(e)$ is offset one-for-one by the equilibrium cutoff
and leaves every child's assignment probability unchanged.
\end{lemma}

\emph{Proof.} Capacity clearing fixes the standardized margin at
$F_\varepsilon^{-1}(1-q)$ in equation~\eqref{eq:symmetric_cutoff}. After a common
change from $e$ to $\widetilde e$, the cutoff
$\widetilde v=v+g(\widetilde e)-g(e)$ restores the same margin for every family,
so every admission probability, and therefore aggregate admissions, is
unchanged. \qed

\medskip

An individual family nevertheless treats the cutoff as fixed. When additional
effort can improve the child's admission chance, a valuable destination induces
positive positional effort: at zero effort, the marginal resource cost is zero
while the marginal private assignment benefit is positive. A planner who preserves
assignments instead sets the common component of this effort to zero because it
uses resources without raising human capital or changing who is admitted.
Lottery-based assignment likewise eliminates the return to positional effort
and sets $e=0$, while productive education remains valuable because it raises
human capital directly. Supplemental Appendix~B.2 states and proves these
results.

The contest can affect fertility through two conceptually distinct channels.
Positional outlays raise the resource cost of every child, while lower per-child
inputs in larger families may reduce each child's assignment probability. Writing
$W_n=np_n$ for the expected number of children receiving the preferred
destination, each additional child adds exactly $q$ to $W_n$ under
lottery-based assignment; the contest therefore lowers the destination-value
contribution to the $n+1$ margin whenever $W_{n+1}-W_n<q$.\footnote{Supplemental
Appendix~B.4 derives the exact condition. The expected number of admitted
children falls with parity only if the per-child admission probability falls
below the fraction $n/(n+1)$ of its previous value: going from one child to two,
the probability would have to fall by more than half.} The analytical
model does not determine which of the two channels is larger, and their
relative size matters for measuring policy effects.

The contest operates only under interior scarcity. Under the additive score
used here, equilibrium positional effort and its resource cost vanish as the
admission share approaches zero or one. Under the quantitative model's
normalized logistic rule, the private marginal return to the score vanishes as
the admission share approaches $q_{\min}$ or one; $q_{\min}$ is close to zero at
the calibrated noise scale. Supplemental Appendix~B.3 derives these capacity
limits. Away from the boundaries, the effect of capacity expansion is
ambiguous because destination value may change with supply. Lottery assignment
is the exception because it
operates directly on the score--allocation link isolated by
Lemma~\ref{lem:absorption}. The quantitative model measures these channels and
compares the policies.

\section{Quantitative Model}
\label{sec:model}

The quantitative model embeds the assignment externality of
Section~\ref{sec:analytical} in an economy built to confront the facts of
Section~\ref{sec:empirical}. It is a stationary overlapping-generations model
with heterogeneous parents who choose fertility, labor supply, productive
education, and positional effort. A fixed share of children is assigned to the
preferred destination through a score contest. A generation lives for one period as an adult. Adults
supply labor, choose a family size, and invest in their children. Children become the
next generation of adults. The stationary distribution is therefore determined by
a fertility-weighted transition operator. An adult enters the period with earnings capacity $h$ and a persistent family-ability
state $\kappa_p$. The state $h$ incorporates the adult's productive human capital
and the wage consequence of the destination received in childhood. The state
$\kappa_p$ governs the conditional distribution of child ability. Read broadly,
the pair stands for permanent family advantage rather than current income alone:
parental education, networks, wealth, and stable employment all shape the
child's prospects in the contest. The joint
distribution at the beginning of the period is $\mu(h,\kappa_p)$.

The within-period timing is:
\begin{enumerate}[label=(\roman*),leftmargin=2em]
\item the adult observes $(h,\kappa_p)$ and the idiosyncratic fertility-choice
shocks;
\item the adult chooses fertility $n\in\{0,1,2,3\}$, accounting for the
expected value of the subsequent participation choice;
\item conditional on $n$, the participation shock is realized and the adult
chooses contest participation $r$, productive education $x$, positional effort $e$, and labor supply
$\ell$;
\item child ability $\kappa'$ is drawn from the transition matrix $\Pi_\kappa$,
where $\Pi_\kappa(\kappa'|\kappa_p)$ is the probability of ability state $\kappa'$
conditional on parental family-ability state $\kappa_p$; the child's human capital
and contest outcome are then determined, and the child enters the next adult
generation.
\end{enumerate}
Education and effort are chosen before child ability is realized, so parents
condition investment on the ability distribution implied by their family
state.

\subsection{Preferences and the Household Problem}
\label{sec:household_problem}

For each parity $n$ and contest-participation state $r\in\{0,1\}$, let
$V_{nr}(h,\kappa_p;v,T)$ denote utility after the household optimizes productive
education $x$, positional effort $e$, and labor supply $\ell$, taking as given
the assignment cutoff $v$ and lump-sum transfer $T$. Participation
means purchasing access to the positional input, not eligibility for a
destination: all children remain in the assignment pool. An entrant pays the
fee $F$ and may choose $e\geq0$; a non-entrant sets $e=0$ but may still invest
in productive education. The branch value is
\begin{align}
V_{nr}(h,\kappa_p;v,T)
=\max_{x,e,\ell}\quad&
\log\!\left(\frac{1.5c_{nr}}{1.5+0.3n}\right)
+\nu\frac{(1-\ell-n\lambda)^{1-\gamma}-1}{1-\gamma}
+\Phi(n) \notag\\
&+\beta n\,\E_{\kappa',D_E|\kappa_p,x,e,v}
\left[
u_c\!\left(y'(x,\kappa',D_E)\right)+B_E D_E
\right]
- n\,\psi\,\frac{\big(x+e\big)^{1+\varphi}}{1+\varphi},
\label{eq:quant_utility}
\end{align}
where $\nu$ and $\gamma$ are the leisure weight and curvature, $\psi$ scales
the child-time cost, $\beta$ is the weight on child outcomes, $D_E$ indicates
receipt of the preferred destination, and $B_E$ is the common value of
destination attributes beyond the current wage. The function
$u_c(y)=(y^{1-\eta_c}-1)/(1-\eta_c)$ is the parent's valuation of child income.
Both education inputs use instructional hours, so the final term prices the
child's time as a function of $x+e$. The curvature $\varphi$ takes the canonical
quadratic value; the Supplemental Appendix reports sensitivity to this choice.

The equivalence scale assigns weight 1 to the first adult, 0.5 to the second
adult, and 0.3 to each child; $\lambda$ is the common time requirement per
child. Direct utility from family size is a free parity-specific level,
\begin{equation}
\Phi(0)=0,\qquad
\Phi(n)=\phi_n,\quad n\in\{1,2,3\},
\label{eq:family_utility}
\end{equation}
with the three levels calibrated jointly with the other parameters. The levels
capture parity-specific preferences and fixed costs and allow the model to match
the concentration of Korean families at two children.

The branch budget constraint is
\begin{equation}
c_{nr}+n\bigl[x+r\,\bigl(F+C(e)\bigr)\bigr]
=wh\ell+T,
\label{eq:quant_budget}
\end{equation}
with $c_{nr}>0$, $0\leq\ell<1-n\lambda$, $x\geq0$, and $e=0$ when $r=0$.
Here $w$ is the wage per efficiency unit of labor, normalized to one.
Positional effort has the resource cost
\begin{equation}
C(e)=c_e\frac{e^{1+\xi_e}}{1+\xi_e}.
\label{eq:effort_cost}
\end{equation}
Here $c_e>0$ scales the cost and $\xi_e>0$ governs its curvature, implying
$C'(0)=0$.
The fee generates an interior participation margin; without it, every household
would purchase some effort.

Conditional on parity, iid type-I extreme-value participation shocks with scale
$\sigma_s$ smooth the choice between the two branches. The participation
probability and its inclusive value are
\begin{align}
\pi_{r|n}(h,\kappa_p)
&=
\frac{\exp\{V_{nr}(h,\kappa_p;v,T)/\sigma_s\}}
{\sum_{j=0}^{1}\exp\{V_{nj}(h,\kappa_p;v,T)/\sigma_s\}},
\label{eq:entry_logit}\\
\overline U_n(h,\kappa_p)
&=\sigma_s\log\!\left[
\sum_{r=0}^{1}\exp\{V_{nr}(h,\kappa_p;v,T)/\sigma_s\}
\right].
\label{eq:entry_inclusive_value}
\end{align}
For $n=0$, only $r=0$ is available, $x=e=0$, and
$\overline U_0=V_{00}$. Fertility-choice shocks $\epsilon_n$ are iid type-I
extreme value with scale $\sigma$. The probability of choosing $n$ children is
\begin{equation}
\Pr(n|h,\kappa_p)
=
\frac{\exp\{\overline U_n(h,\kappa_p)/\sigma\}}
{\sum_{j=0}^{3}\exp\{\overline U_j(h,\kappa_p)/\sigma\}}.
\label{eq:fertility_logit}
\end{equation}
Let $\pi_{nr}(h,\kappa_p)=\Pr(n|h,\kappa_p)\pi_{r|n}(h,\kappa_p)$ denote the
joint probability of parity and participation, with $\pi_{0|0}=1$. The two
nested logits represent unobserved household heterogeneity in participation and
fertility.

\subsection{Child Human Capital and the Assignment Contest}
\label{sec:child_human_capital}

Child ability follows a discretized AR(1) process:
\begin{equation}
\log\kappa'
=\rho_\kappa\log\kappa_p+\eta',
\qquad \eta'\sim N(0,\sigma_\kappa^2).
\label{eq:ability_process}
\end{equation}
The process is discretized with $N_\kappa$ states using the method of
\citet{Tauchen1986}, yielding the transition matrix $\Pi_\kappa$. Child human
capital is
\begin{equation}
h'
=A_h\kappa'\left(\theta_0+x^{\alpha_1}\right),
\qquad 0<\alpha_1<1.
\label{eq:human_capital}
\end{equation}
The productivity constant $A_h$ is normalized to one. The baseline component
$\theta_0$ summarizes publicly provided education and other human capital
produced without private education spending. Productive education
$x$ raises human capital in every destination.

If the child receives the preferred destination, adult earnings capacity is
multiplied by $\exp(\delta_E)$:
\begin{equation}
\log y'
=\log(wh')+\delta_E D_E.
\label{eq:child_income}
\end{equation}
The current-wage premium $\delta_E$ therefore affects both the parent's expected
child value and the next generation's realized income distribution. By contrast,
$B_E$ affects parental utility but is not recorded as current earnings. This
separation lets the model match the measured wage distribution while allowing
families to value destination attributes outside the contemporaneous wage measure.

\label{sec:assignment_contest}

Productive human capital and positional effort determine a positive contest score,
\begin{equation}
S(h',e)
=(h')^{\theta_h}(1+e).
\label{eq:contest_score}
\end{equation}
The parameter $\theta_h$ governs the role of productive human capital in assignment.
Effort $e$ affects the score but not $h'$ and changes income only through the
destination indicator. Taking logs yields the additive-score counterpart of
Section~\ref{sec:analytical}; the logistic assignment rule below smooths the
sharp cutoff.

The probability of receiving the preferred destination is a smooth approximation to
a cutoff rule:
\begin{equation}
P_E(S;v)
=\Lambda\!\left(\frac{S-v}{\sigma_s v}\right),
\qquad
\Lambda(z)=\frac{1}{1+\exp(-z)}.
\label{eq:admission_probability}
\end{equation}
Here $v>0$ is the equilibrium cutoff and $\sigma_s$ controls residual assignment
noise. A smaller $\sigma_s$ produces a sharper rank-order rule. The smooth
probability reflects idiosyncratic admission noise and is convenient
numerically. The scale $\sigma_s$ is the same one that smooths participation.
One unobserved trait motivates the shared scale: how well the child is suited
to the examination race. Parents observe part of that suitability when
deciding whether to enter, which disperses entry; the remainder appears as
noise in assignment. The cutoff clears a fixed capacity: the birth-weighted share of admitted
children equals the admission share $q_E$, the population share of the
preferred destinations measured in
Section~\ref{sec:preferred_destinations}. Writing $s=(h,\kappa_p)$ and letting
$x_{nr}(s)$ and $e_{nr}(s)$ denote the branch policies, the capacity condition is
\begin{equation}
\frac{
\int\sum_{n=1}^{3}\sum_{r=0}^{1}
\pi_{nr}(s)n
\E_{\kappa'|\kappa_p}
\left[P_E\!\left(S(h'(x_{nr}(s),\kappa'),e_{nr}(s));v\right)\right]
\dd\mu(s)
}{
\int\sum_{n=1}^{3}\sum_{r=0}^{1}\pi_{nr}(s)n\dd\mu(s)
}
=q_E.
\label{eq:quant_capacity}
\end{equation}
All children face the same cutoff. Participation determines access to positional
effort, not whether a child can receive the destination.

The expected child-value term in equation~\eqref{eq:quant_utility} can be written
\begin{align}
\E_{\kappa'|\kappa_p}
\Bigl[
u_c(wh')
+P_E(S;v)\,\bigl\{u_c\!\bigl(wh'\exp(\delta_E)\bigr)-u_c(wh')+B_E\bigr\}
\Bigr].
\label{eq:expected_child_value}
\end{align}
With CRRA curvature, the utility value of the wage premium varies with child
income, while $B_E$ is additive.\footnote{In the logarithmic limit
$\eta_c\rightarrow1$ the assignment term becomes $(\delta_E+B_E)P_E$, the form
used in Section~\ref{sec:analytical}.} The parameter $B_E$ is common across
parents and does not depend on the rank of the child or the outcomes of other
children. Both components raise the private return to effort.

\subsection{Equilibrium, Policy, and Counterfactuals}
\label{sec:lottery_counterfactual}
\label{sec:stationary_distribution}
\label{sec:model_policies}

Conditional on the parental state, parity, ability realization, and destination
outcome, a child enters the next adult generation with the state
\begin{equation}
\widetilde h'
=h'\exp(\delta_E D_E),
\qquad \kappa_p'=\kappa'.
\label{eq:next_state}
\end{equation}
Let $\mathcal{R}$ be the reproduction operator. For each parental state, it
counts the children born, $\sum_{n,r}\pi_{nr}(h,\kappa_p)n$, and distributes them
over next-generation states according to the ability transition, the
human-capital policy, and the assignment probability. Because families differ
in size, $\mu\mathcal{R}$ is a measure of births rather than a probability
distribution, and its total mass is births per adult. A stationary
distribution reproduces its own composition once normalized by that mass:
\begin{equation}
\mu'=\frac{\mu\mathcal{R}}{\int \dd(\mu\mathcal{R})}=\mu.
\label{eq:stationarity}
\end{equation}
The population may shrink across generations. Stationarity fixes the composition
of each generation: $\mu$ is the normalized nonnegative left eigenvector of
$\mathcal{R}$ associated with its Perron root. That root equals births per adult
in the stationary composition and governs the change in cohort size.

This construction generates the income Gini and intergenerational income
persistence jointly. Productive investment, inherited ability, and destination
assignment all affect next-generation income. The preferred-destination cutoff must
therefore be solved together with the stationary distribution.

\begin{definition}[Stationary competitive equilibrium]
\label{def:equilibrium}
A stationary competitive equilibrium consists of branch policy functions,
participation and fertility probabilities, a cutoff $v$, a lump-sum transfer
$T$, and a distribution $\mu$ such that: (i) household policies solve the branch
problems and the two nested logits; (ii) the preferred-destination market satisfies
equation~\eqref{eq:quant_capacity}; (iii) the government budget balances; and
(iv) $\mu$ satisfies equation~\eqref{eq:stationarity}.
\end{definition}

The counterfactuals alter one element of this equilibrium and re-solve the
rest. The main experiment replaces score-based assignment with a
capacity-preserving lottery, which sets
\begin{equation}
P_E(S;v)=q_E
\label{eq:lottery_probability}
\end{equation}
for every child. The destination wage premium and $B_E$ remain in expected child
value, and a fraction $q_E$ of children continues to receive the wage premium in the
next-generation distribution. Because neither $x$ nor $e$ can change assignment
under the lottery, the household sets $e=0$. Productive education remains positive
because it raises $h'$. The prize counterfactuals separate the measured and
unmeasured components. Setting $B_E=0$ retains only the wage premium; setting
$\delta_E=0$ retains only the common
destination value; and setting both to zero eliminates the prize. Each experiment
holds the remaining parameters fixed and re-solves the equilibrium.

The fiscal experiments introduce a tax on total private-education spending, an
oracle tax on positional outlays alone, and a cash transfer per child. Let
$\tau$ denote the broad spending tax and $\tau^{P}$ the oracle positional tax.
The budget constraint becomes
\begin{equation}
c_{nr}+n\Bigl[(1+\tau)x
+r(1+\tau+\tau^{P})\bigl(F+C(e)\bigr)\Bigr]
=wh\ell+T+nb,
\label{eq:policy_budget}
\end{equation}
where $b$ is the pronatal transfer. The broad tax applies to both components of
observed private-education spending. The oracle sets $\tau=0$ and applies
$\tau^{P}$ only to $F+C(e)$; it is a diagnostic benchmark because the productive
and positional components are not separately observed. Government budget
balance requires
\begin{equation}
T+\int\sum_{n=1}^{3}\sum_{r=0}^{1}\pi_{nr}(s)nb\,\dd\mu(s)
=\int\sum_{n=1}^{3}\sum_{r=0}^{1}\pi_{nr}(s)n
\Bigl[\tau x_{nr}(s)
+r\,(\tau+\tau^{P})\bigl(F+C(e_{nr}(s))\bigr)\Bigr]\dd\mu(s).
\label{eq:government_budget}
\end{equation}
The transfer $T$ can be negative when the per-child payment exceeds tax revenue.
The revenue-funded package instead fixes $\tau=0.20$ and $T=0$, choosing $b$ so
that education-tax revenue equals aggregate child-transfer spending.
Every policy is solved as a new stationary equilibrium, including the cutoff and
the distribution.
Supplemental Appendix~C gives the complete conditional household problem, the
reproduction matrix, and the equilibrium algorithms for the benchmark, the full
lottery, the qualified lotteries of Section~\ref{sec:qualified_lottery}, and the
fiscal experiments.

\section{Benchmark Calibration and Lottery Counterfactual}
\label{sec:calibration}
\label{sec:results}

I calibrate the quantitative model for the 1970--75 cohort as a benchmark and
report the main counterfactual. I first fix seven parameters using direct
measurements or standard conventions. KLIPS wave 27 gives the population share
of preferred destinations, $q_E=0.1845$, and their conditional current-wage
premium, $\delta_E=0.257$. KTY's KLIPS measure of parental time per child gives
the common child time requirement, $\lambda=0.041$ of the weekly endowment, or
5.7 hours per week. The other four parameter values follow standard conventions.
The child-time cost is quadratic ($\varphi=1$), the positional effort cost has
curvature $\xi_e=1.5$, and the leisure and child-income curvatures both equal
two ($\gamma=\eta_c=2$, the inverse elasticity of substitution used by
KTY).\footnote{The Parameter Sensitivity section of Supplemental Appendix~D
profiles the child-income curvature, varies the positional and child-time cost
curvatures, and checks the education and earnings grids. The lottery effect
remains large in every converged alternative.}

I calibrate the remaining sixteen parameters jointly by minimizing a
minimum-distance criterion based on ten point targets and two interval
restrictions. I use the parameter vector with the lowest criterion value obtained
in the numerical search as the benchmark. Since the model moments depend jointly
on the parameter vector, there is no one-to-one link between individual parameters
and moments. Some moments are nevertheless more informative for particular
parameters, as the discussion below explains, even though there is no formal
identification procedure.

For the ten point targets, I express the residuals as relative deviations,
placing moments measured in different units on a common percentage scale. The
other two restrictions are KEEP estimates: the wage return to private-education
spending and the logit slope of preferred-destination entry with respect to log
spending. They enter with zero loss inside their ninety-five percent confidence
intervals. Supplemental Appendix~D reports the criterion and parameter
transformations. The lottery counterfactual is computed only after the calibration
is complete. It replaces score-based assignment with a lottery that preserves
capacity and destination value. Intermediate regimes locate the source of the
fertility response, and prize counterfactuals identify which part of destination
value sustains the race.

\subsection{Calibration Results}
\label{sec:moments}

Most of the calibration moment set follows KTY, but all spending objects use a
common private-education spending basis. Table~\ref{tab:moments} lists the twelve:
the three parity shares, the income Gini, work hours, the fertility--income
elasticity, bottom-quintile childlessness, the lifetime private-education spending
level, intergenerational income persistence, the private-education burden
gradient, and
the two KEEP return estimates. The lifetime private-education target is 6.4
percent of income, or 72.8 percent of the 8.8 percent total in
Section~\ref{sec:empirical}. This is the relevant accounting base because
publicly provided education is common across households, whereas the model's
household-specific inputs correspond to private-education spending. This accounting choice
excludes KTY's total-spending investment--income elasticity from the target set.\footnote{A private-education version
of the elasticity would largely duplicate the burden gradient, which already
disciplines how spending varies with income.} A model analogue on the same spending basis is
reported as a descriptive statistic in Table~\ref{tab:moments}.

Three moments directly discipline the contest: the private-education burden
gradient across income quintiles, the KEEP wage return to private-education spending, and
the logit slope of preferred-destination entry with respect to log spending. The
first determines how the cost of the contest varies with income. The two KEEP
moments restrict the productive return and the assignment margin. Their estimates
are $0.0017$ (s.e. $0.0125$) and $0.027$ (s.e. $0.059$), yielding intervals of
$[-0.023,0.026]$ and $[-0.088,0.142]$. Unlike the other ten moments entering as point targets,
these two KEEP coefficients have confidence intervals that include zero, so even
their signs are uncertain. Point targets would pull the model toward sampling
noise. The interval restrictions require only that the model-implied coefficients
lie within the ranges that the data cannot reject.

\begin{table}[!t]
\centering
\footnotesize
\setlength{\tabcolsep}{6pt}
\caption{Internally calibrated parameters}
\label{tab:parameters}
\begin{threeparttable}
\begin{tabular}{llr}
\toprule
Parameter & Interpretation & Value \\
\midrule
$\phi_1,\phi_2,\phi_3$ & Parity preference levels & $1.469,\ 2.329,\ -1.658$ \\
$\beta$ & Weight on children & 1.801 \\
$\sigma$ & Fertility-choice dispersion & 1.841 \\
$\nu$ & Leisure weight & 1.558 \\
$\sigma_\kappa,\rho_\kappa$ & Ability dispersion, persistence & $0.450,\ 0.302$ \\
$B_E$ & Destination value beyond current wages & 5.35 \\
$c_e$ & Positional cost scale & 0.050 \\
$\sigma_s$ & Gate noise and participation dispersion (one parameter) & 0.240 \\
$\theta_h$ & Score loading on human capital & 0.242 \\
$\alpha_1$ & Education elasticity in human-capital production & 0.035 \\
$\theta_0$ & Human-capital intercept & 0.969 \\
$\psi$ & Price of the child's time & 2.429 \\
$F$ & Contest entry fee (share of income) & 0.019 \\
\bottomrule
\end{tabular}
\end{threeparttable}
\end{table}

Table~\ref{tab:parameters} reports the sixteen calibrated parameters. The calibrated
education elasticity is $\alpha_1=0.035$, far below the value of
0.4 commonly used in quantity--quality models, including KTY. Fixing it at 0.4
pushes the implied wage return outside the KEEP interval and lowers
private-education spending to 5.0 percent of income against the 6.4 percent target. It also
raises the root mean squared error across the nine baseline moments from 0.7 to
13.3 percent.
The destination value $B_E=\vBE$ follows from the same accounting. Because the
model-implied productive return is small, the model can sustain the observed level of
spending only through the assignment value of preferred destinations. The calibration
therefore assigns those destinations substantial value beyond their current-wage premium.

\begin{table}[!b]
\centering
\footnotesize
\setlength{\tabcolsep}{5pt}
\caption{Moments and model fit}
\label{tab:moments}
\begin{threeparttable}
\begin{tabular}{lrrr}
\toprule
Moment & Model & Measured & Error \\
\midrule
\multicolumn{4}{l}{\emph{Panel A: baseline moments (KLIPS, following KTY)}} \\
One-child share & 0.201 & 0.203 & $-0.6$\% \\
Two-child share & 0.613 & 0.625 & $-1.9$\% \\
Three-or-more share & 0.140 & 0.140 & $-0.4$\% \\
Income Gini & 0.264 & 0.263 & $+0.2$\% \\
Work-time share & 0.292 & 0.292 & $-0.0$\% \\
Fertility--income elasticity & 0.079 & 0.079 & $-0.0$\% \\
Q1 childlessness & 0.065 & 0.065 & $-0.2$\% \\
Private-education spending / income & 0.064 & 0.064 & $+0.0$\% \\
Income elasticity of private-education spending & 0.813 & 0.788 & untargeted \\
Intergenerational income persistence & 0.320 & 0.320 & $+0.0$\% \\
\addlinespace
\multicolumn{4}{l}{\emph{Panel B: contest moments (KLIPS, KEEP)}} \\
Education burden, Q1/Q5 (private) & 1.297 & 1.297 & $-0.0$\% \\
Wage return to spending & 0.0171 & $[-0.023,\,0.026]$ & inside \\
Destination-entry logit slope & 0.0101 & $[-0.088,\,0.142]$ & inside \\
Household participation & 0.808 & 0.766 & untargeted \\
Participation, Q5/Q1 & 1.461 & 1.486 & untargeted \\
\bottomrule
\end{tabular}
\end{threeparttable}
\begin{fullwidthtablenotes}
The RMSE covers all ten point targets: nine baseline moments in Panel A
and the burden gradient in Panel B. The
income elasticity of private-education spending and both participation rows are untargeted.
The elasticity uses within-year income quintiles in the data and
permanent-income quintiles in the model; participation compares positive
private-education spending with contest entry. The burden gradient is a point target;
the KEEP moments enter as ninety-five percent intervals.
\end{fullwidthtablenotes}
\end{table}

Table~\ref{tab:moments} shows that the model fits all ten point targets very
closely. Their root mean squared relative error is \vPointRmse\ percent, and
no error exceeds two percent.\footnote{Using the same percentage-RMSE
formula, KTY's published benchmark has an RMSE of about 8.6 percent across the
ten comparable point targets. The moment sets differ: KTY targets the total-spending
investment--income elasticity, whereas this paper targets the private-education
burden gradient.} Both KEEP interval restrictions are also satisfied. The table also reports three
untargeted model implications that closely match their empirical counterparts.
The income elasticity of private-education spending is 0.81 in
the model and 0.79 in KLIPS. Contest
participation is not directly observed in the data. KLIPS instead measures a
related event: positive private-education spending in a household-year. The
calibrated entry margin is consistent with that measurement: the contest-entry
rate is 0.81 against the any-spending rate of 0.77, and the entry gradient
across income quintiles differs by about two percent from the measured participation
gradient. The private-education spending level and burden gradient in the
objective discipline these untargeted margins.

The calibration assigns \vPosSharePct\ percent of private-education spending to
positional effort. This share is an outcome of the calibration, not a target.
Purpose codes in the Private Education Expenditures Survey (PEES) place an
accounting upper bound of 0.97: only 3 percent of
spending is reported as supervision or sociality.\footnote{Supplemental Appendix~D reports the
detailed PEES purpose-code classification.} Thirty percent of general-subject
spending is explicitly motivated by admission preparation, anxiety, or advance
study and is therefore clearly competition-related. Another 66 percent is
reported as supplementing or deepening school lessons. This spending may build
human capital, protect a child's rank, or do both. The calibration uses the KEEP
return estimates, together with the level and income gradient of
private-education spending, to determine how much of this component is
productive and how much is positional.
The result closely matches independent structural evidence from
\citet{Kang2026}. Estimating a dynamic Korean admissions tournament, he finds
that competition accounts for 76.6 percent of tutoring expenditure over
secondary school and 89 percent in the final year. Tutoring in his model can also build
skills, whereas the positional input here cannot; the two models nevertheless
assign nearly the same share of tutoring expenditure to the admissions race.

The small productive return is not imposed by the KEEP wage restriction.
Dropping that restriction and recalibrating leaves $\alpha_1$ at 0.035 and the
implied wage return at 0.017. The spending level and burden gradient therefore
select the low productive return; KEEP corroborates rather than creates that
result.

\begin{remark}[Spending spillover]
KTY estimate that a reduction in top families' private-education spending lowers
spending among other families, with a spillover coefficient of $0.039$ (s.e.\
$0.015$). They use this moment to calibrate the comparison externality in
their model. It is not a calibration target here because the assignment model has
no common reference level. I reproduce KTY's design while allowing the 11~p.m.
curfew to affect other families directly. In the lower-education sample used for
KTY's calibration, the estimate is then $0.036$ (s.e.\ $0.025$), with an
Anderson--Rubin confidence set of $[-0.022,\,0.090]$.\footnote{Conditional on the
10~p.m. rule, the 11~p.m. rule has little independent first-stage power. With the
11~p.m. rule included as a control rather than a second instrument, five of the
six sample--specification confidence sets include zero; the exception is the
baseline lower-income sample.}
Under the relaxed design, the data therefore do not require a positive
common-reference channel, although the point estimates remain positive.
\end{remark}

\begin{remark}[The mechanisms on a common utility scale]
At the benchmark distribution, KTY's comparison term and this paper's admission
prize can be expressed on the same utility scale. KTY's reference income, the
mean among the top 15 percent of children, is 4.35, compared with an overall
mean of 2.26. Their calibrated comparison weight lowers utility at mean income
by 0.078 and raises the marginal utility of child income by a factor of 1.39.
Here the expected admission prize is 1.01, or 12.9 times that utility
subtraction. KTY's mechanism therefore moderately raises every family's
marginal return, whereas the assignment mechanism offers a large discrete prize
with probability 0.1845.
\end{remark}

\subsection{The Lottery Counterfactual}
\label{sec:main_counterfactual}

The main counterfactual replaces score-based allocation with a
capacity-preserving lottery. In the benchmark, children are ranked by their
scores and a share $q_E$ receives a preferred destination. Under the lottery,
every child receives the same admission probability $q_E$, regardless of score.
The same share of children therefore reaches preferred destinations, but family
choices no longer determine who does. The wage premium $\delta_E$, the common
nonwage value $B_E$, and the productive return to education also remain
unchanged. Productive education continues to build human capital, but neither
productive nor positional spending can improve admission.

I use the lottery as the main experiment because it removes the institutional source
of the assignment externality without changing the scarcity or value of the
prize. Transfers change family resources, education taxes change the price of
investment, and capacity expansion changes the supply and potentially the value
of preferred destinations. The lottery instead removes the admission advantage
that families can obtain through spending. It therefore measures the fertility
cost of score-based competition itself. Completed fertility rises from
\vContest\ to \vLottery, an increase of \vEffect\ children. The response is
concentrated at higher parities. The share of families with
three or more children rises from 14.0 percent in the contest benchmark to 23.9 percent
under the lottery, while the two-child share falls. The largest change occurs at
the two-to-three-child margin.

\begin{table}[!b]
\centering
\footnotesize
\setlength{\tabcolsep}{2.5pt}
\caption{Assignment regimes}
\label{tab:main_result}
\begin{threeparttable}
\begin{tabular}{lrrrrrr}
\toprule
Regime & Fertility & Effect & \shortstack{Productive\\$x/y$} & \shortstack{Contest\\outlay $/y$} & \shortstack{Entry\\rate} & \shortstack{Output per\\unit of labor} \\
\midrule
1. Benchmark contest & 1.846 & -- & 0.015 & 0.050 & 0.808 & 1.000 \\
2. Infeasible frontier: sorting retained & 2.082 & $+0.235$ & 0.030 & 0.000 & 0.000 & 1.010 \\
3. Capacity-preserving lottery & 2.086 & $\mathbf{+0.240}$ & 0.017 & 0.000 & 0.000 & 0.993 \\
\bottomrule
\end{tabular}
\end{threeparttable}
\begin{fullwidthtablenotes}
All rows hold $q_E$, $\delta_E$, and $B_E$ fixed. Row~2 is an infeasible
frontier: it removes the positional-preparation entry branch, so no
household pays the fee or chooses positional effort, while retaining
score-based assignment through productive human capital.
Row~3 assigns every child an admission probability of $q_E$, so positional effort and
contest entry disappear. Productive education and contest outlay are
shares of household income. Output per unit of labor is
$\E[wh\ell]/\E[\ell]$, normalized to one in Regime~1. The entry rate
is the share of households paying the fee.
\end{fullwidthtablenotes}
\end{table}

Table~\ref{tab:main_result} compares the lottery counterfactual (Regime~3) with the
current score-based benchmark (Regime~1). In the benchmark, households devote 1.5 percent of income to
productive education and \vRaceMoneyPct\ percent to positional spending. Under
the lottery, productive education rises slightly to 1.7 percent, positional
spending disappears, and output per unit of labor falls by only
\vLotteryProdCostBenchmark\ percent. Random assignment sacrifices some positive
sorting between human capital and preferred destinations, but eliminating an
input that builds no human capital releases enough resources to preserve
productive investment. The result is a large fertility increase with a modest
productivity cost.

Regime~2 is an infeasible frontier, not a policy counterfactual. It removes the
positional-preparation entry branch, so no household pays the fee or chooses
positional effort, while retaining assignment through productive human capital.
An admission system observes the resulting score, not the two inputs separately,
and therefore cannot implement this allocation directly. The regime nevertheless
isolates the fertility contribution of the positional channel while retaining
productive sorting: fertility rises to \vFrontier,
accounting for \vFrontierShare\ percent of the lottery effect, while output per
unit of labor remains above the current benchmark.
Moving from this frontier to the lottery adds almost no fertility and lowers
output per unit of labor by \vLotteryProdCost\ percent.

\label{sec:channels}
\label{sec:admission_by_parity}

The fertility response comes from the resources released when the positional
channel shuts down. The average admission prospect remains $q_E$, but families no longer
spend resources to improve their relative scores. The benchmark's positional
outlay disappears, productive education is maintained, and the recovered
resources support more children.
A second channel is dilution within the family. Another child may spread
education spending more thinly and reduce the admission probability of each
child. Quantitatively, this force is small. Holding the household's state fixed,
the admission probability per child changes by less than six percent of its
level between one and three children, and the expected number of admitted
children rises by almost $q_E$ with each additional child. Raw admission rates
actually rise with family size, from 0.164 for one-child families to 0.197 for
two-child families and 0.223 for families with three or more children. This is
selection: families choosing more children tend to occupy states with better
admission prospects. The fixed-state comparison isolates the dilution faced by
a family deciding whether to have another child.

Education competition therefore raises the cost of another child mainly because
contest spending grows with family size, not because an additional sibling
sharply reduces the admission prospects of existing children. This distinction
matters for capacity policy: adding preferred positions may draw more families
into the contest even as it relaxes scarcity, as Section~\ref{sec:capacity}
shows.

\label{sec:value_results}

The size of this resource channel depends on the prize. A contest that absorbs
five percent of household income requires preferred destinations valuable enough
to justify the outlay. The model separates the measured current-wage premium
$\delta_E$ from $B_E$, a composite destination value beyond current wages.
Table~\ref{tab:prize_decomposition} varies these two components while holding the
remaining parameters fixed and re-solving both assignment regimes.

\begin{table}[htbp]
\centering
\footnotesize
\caption{Destination-value decomposition}
\label{tab:prize_decomposition}
\begin{threeparttable}
\begin{tabular}{lrrrr}
\toprule
Destination value & \shortstack{Contest\\fertility} & \shortstack{Lottery\\fertility} & Difference & \shortstack{Contest outlay\\ / income} \\
\midrule
Benchmark: $\delta_E+B_E$ & 1.846 & 2.086 & $+0.240$ & 0.050 \\
Current wage only: $B_E=0$ & 1.614 & 1.608 & $-0.006$ & 0.015 \\
Large-firm mover premium only & 1.605 & 1.597 & $-0.007$ & 0.015 \\
Unmeasured value only ($\delta_E=0$) & 1.842 & 2.075 & $+0.233$ & 0.051 \\
No destination value & 1.600 & 1.592 & $-0.008$ & 0.015 \\
\bottomrule
\end{tabular}
\end{threeparttable}
\begin{fullwidthtablenotes}
$\delta_E=0.257$ is the measured composite preferred-destination
wage premium in KLIPS. The large-firm mover row uses the event-study
estimate of 0.076 documented in Supplemental Appendix~A and sets $B_E=0$.
Each row holds the remaining benchmark
parameters fixed.
\end{fullwidthtablenotes}
\end{table}

Measured wage premiums alone cannot sustain the benchmark race. With $B_E=0$,
contest outlay falls from 0.050 to 0.015 of income and the lottery effect
disappears: it is $-0.006$ with the current-wage premium and $-0.007$ with the
large-firm mover estimate.\footnote{The ``Preferred-Destination Measurement''
subsection of Supplemental Appendix~A reports the large-firm mover sample,
event-study specification, and estimate.} By
contrast, retaining $B_E$ while setting
$\delta_E=0$ leaves contest outlay at 0.051 and the lottery effect at $+0.233$,
close to \vEffect. The race and its fertility effect are therefore sustained by
destination value beyond current wages. $B_E$ collects employment security, career progression, deferred compensation,
professional standing, and social recognition that a current-wage regression
does not measure. \citet{FangLiu2026} emphasize the social value attached to
examination achievement as one component of this broader prize. The model does
not price these attributes separately. Given the KEEP bounds on the productive
return, $B_E$ is the composite value required to reconcile observed spending
with competition for preferred destinations.

The prize cannot be arbitrarily small or large while the remaining parameters
stay at their benchmark values. Setting $B_E=2$ lowers private-education spending to 4.5
percent of income and participation from 0.81 to 0.64, below the 6.4 percent
spending target and the 0.77 participation measurement. Setting $B_E=8$ instead
raises private-education spending to 8.0 percent and participation to 0.86, overshooting
both. These probes illustrate why $B_E$ is jointly rather than separately
disciplined: changing it alone moves both spending and participation away from
their empirical counterparts.

\section{Policy Experiments}
\label{sec:policy}
\label{sec:policy_assignment}

The benchmark counterfactual shows that the assignment rule is the central
policy margin: a full lottery removes the return to positional spending while
leaving the admission share and value of preferred destinations unchanged. A
full lottery is an institutional benchmark and would be difficult to implement
as a replacement for achievement-based admission in practice. The relevant policy
question is how much of the fertility gain survives when admission continues to
reward achievement, and whether more familiar interventions can reproduce it
without changing the assignment rule. \citet{MahlerTertiltYum2025} likewise
discuss adding randomness to selection, expanding elite-college capacity, and
reducing the salience of high-stakes tests and precise rankings as ways to soften
admissions competition.

I begin with a qualified lottery, which uses the score only to determine who
enters a qualified pool and then draws the available positions at random
within that pool. Varying the size of the pool traces a path from the current
score ranking to random assignment and reveals how much score differentiation
must be removed before the race weakens. I then consider capacity expansion and
fiscal policies. Capacity expansion increases the number of preferred
positions but preserves score-based assignment, whereas transfers and education
taxes change household resources or the price of spending without changing
scarcity. These comparisons show whether the assignment reform can be
replicated by adding positions or subsidizing children, and why a modest change
in the admission rule can initially intensify the contest. The section ends by
comparing the welfare and productivity consequences of all the experiments.
All fiscal policies are evaluated in balanced-budget stationary equilibria.

\subsection{The Qualified Lottery}
\label{sec:qualified_lottery}

A qualified lottery retains the examination as a screening device but stops
using fine differences in scores to rank students within a broad qualified
pool. \citet{FangLiu2026} present this rule as a potential alternative to
score-based admission. Randomization within the qualified pool preserves an
achievement threshold while reducing the payoff to marginal investments in
rank. They argue that the resulting decline in education competition lowers
the expected cost of a child and can therefore raise fertility. Following
their proposal, let $m$ denote the population share above the qualification
threshold. The economy still has positions for a share $q_E$ of the cohort, so
each qualifier is admitted with probability $q_E/m$.
When $m=q_E$, only the top $q_E$ share of the cohort qualifies and every
qualifier is admitted; this is the current score-based assignment. Intermediate
values preserve a score threshold while progressively reducing the return to
rank above it. At $m=1$, everyone qualifies and assignment is random. The
qualified lottery then coincides with the full-lottery counterfactual:
positional effort, contest entry, and the associated fee disappear.
Supplemental Appendix~C gives the assignment probability and equilibrium
algorithm for intermediate values of $m$.

\begin{table}[htbp]
\centering
\footnotesize
\setlength{\tabcolsep}{6pt}
\caption{The qualified lottery: qualify by score, then draw}
\label{tab:qualified}
\begin{threeparttable}
\begin{tabular}{lrrrr}
\toprule
Qualified share $m$ & \shortstack{Fertility\\change} & \shortstack{Share of full\\lottery effect} & \shortstack{Contest\\outlay} & \shortstack{Participation\\rate} \\
\midrule
0.184 (benchmark contest) & +0.0000 & 0\% & 0.0497 & 0.808 \\
0.25 & -0.0080 & -3\% & 0.0524 & 0.829 \\
0.35 & -0.0003 & 0\% & 0.0549 & 0.869 \\
0.50 & +0.0361 & 15\% & 0.0549 & 0.919 \\
0.70 & +0.1115 & 46\% & 0.0480 & 0.952 \\
0.85 & +0.1807 & 75\% & 0.0406 & 0.937 \\
1.00 (full lottery) & +0.2399 & 100\% & 0.0000 & 0.000 \\
\bottomrule
\end{tabular}
\end{threeparttable}
\begin{fullwidthtablenotes}
The threshold clears so the qualified share of the population equals $m$; every qualifier is admitted with the uniform probability $q_E/m$, so expected admissions equal $q_E$ in every row. At $m=q_E$ the rule is the benchmark contest; at $m=1$ everyone qualifies, the cutoff is immaterial, and the experiment coincides with the full lottery.
\end{fullwidthtablenotes}
\end{table}

The response is non-monotone because widening the pool changes both the number
of contestants and the value of an additional score. Just above $m=q_E$, the
reform gives families outside the original admission range a realistic chance
of qualifying. Entry and spending can therefore rise before the link between
scores and admission becomes weak enough to shrink the race. At $m=0.25$, each
qualifier has roughly a three-in-four admission chance, participation rises
from 0.81 to 0.83, and fertility falls by 0.008. Participation reaches 0.95 at
$m=0.70$, even as the fertility gain begins to emerge. The same threshold effect
appears in \citet{Kang2026}'s estimated Korean
admissions tournament, in which randomizing assignment only above a college-tier cutoff
reduces tutoring expenditure by only 7 to 11 percent, because families continue to
compete to reach the qualified pool.

The rule becomes effective only when qualification is broad. A pool containing
half the population delivers 15 percent of the full-lottery fertility effect,
70 percent qualifying delivers 46 percent, and 85 percent delivers 75 percent.
Universal qualification delivers the full effect. The exercise
therefore identifies the operative design margin: preserving an exam threshold
is compatible with a large fertility gain, but only if the threshold is broad
enough that small score differences among qualified students no longer
determine who receives the preferred destinations.

\subsection{Capacity and Its Interaction with Assignment}
\label{sec:capacity}

Adding preferred positions appears to be the most direct way to relax scarcity.
However, it need not weaken the education race. More positions make a preferred
destination attainable for more families and can draw new contestants into the
race. At the same time, a larger supply may make each position less scarce and
therefore less valuable. Whether capacity raises fertility depends on which of
these forces dominates. Let $q$ denote the share of the cohort receiving a
preferred destination and $q_0$ its benchmark value. I write
$B_E(q)=B_E(q_0)(q/q_0)^{-\eta}$, where $\eta$ is the absolute elasticity of the
nonwage destination value with respect to its supply. When
$\eta=0$, a new position carries the same value as an existing one, so expansion
creates additional aggregate destination value. When $\eta=1$, the nonwage
value of each position falls one-for-one with supply, leaving the aggregate
nonwage prize $qB_E(q)$ unchanged. The measured wage premium $\delta_E$ is held
fixed. Because its utility value is small relative to $B_E$, $\eta$ governs most
of the change in the incentive to compete.

The results show why capacity alone is not necessarily pronatal. When
$\eta=1$, expansions of 2.5 to 30 percent lower fertility by $0.001$ to $0.009$.
The reform places the prize within reach of more families without creating much
new aggregate value, so additional entry initially intensifies the score-based
contest. Large expansions eventually make access sufficiently common to relax
scarcity, but the gain still depends heavily on what the new positions are
worth. Raising the preferred-destination share to one half increases fertility
by 0.269 when every new position retains the benchmark value ($\eta=0$), but by
only 0.043 when expansion mainly divides a fixed nonwage prize ($\eta=1$).

\begin{figure}[t]
\centering
\includegraphics[width=0.82\textwidth]{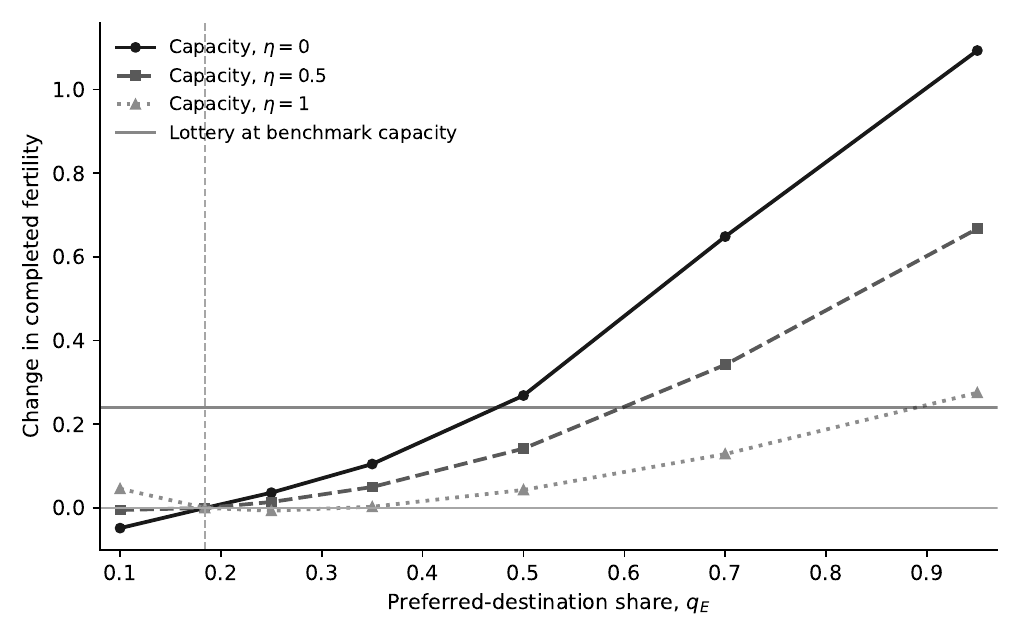}
\caption{Capacity expansion under score-based assignment. Each line reports
the change in completed fertility as the preferred-destination share $q_E$
changes. The elasticity $\eta$ is the absolute supply elasticity of each
position's nonwage value: lower values mean that expansion creates more
aggregate destination value. The horizontal line is random assignment at the
benchmark capacity, and the vertical line marks that capacity.}
\label{fig:capacity}
\end{figure}

The assignment rule matters at every capacity reported in Supplemental
Appendix~E. For every $\eta$, random assignment raises fertility relative to
score-based assignment at the same $q_E$. At $\eta=1$ and $q_E=0.25$, for
example, capacity alone lowers
fertility by 0.007, whereas the combined reform raises it by 0.244. The gap
narrows as $q_E$ approaches one because near-universal access leaves little
incentive to compete for rank under either assignment rule.

\FloatBarrier
\subsection{Transfers and Education Taxes}
\label{sec:transfers}
\label{sec:education_taxes}

Transfers and education taxes change the household budget or the price of
spending, but leave the admission rule intact. Families therefore retain the
same reason to outspend one another for a preferred destination. This is why
policies that look powerful in a household budget have little effect once the
education race adjusts in equilibrium.

A balanced-budget pronatal transfer directly lowers the net resource cost of a
child. A transfer equal to two percent of the average wage per child raises
completed fertility by \vTransferTwo, and a five-percent transfer raises it by
\vTransferFive. The larger transfer delivers only about one eighth of the
assignment effect because it subsidizes children without removing the
competitive spending attached to them. This result gives quantitative content
to \citet{FangLiu2026}'s argument that cash transfers are weak against an
equilibrium of competitive parental investment.

Education taxes attack the spending more directly, but not the incentive that
produces it. Because productive and positional outlays are not separately
observed, the policy-relevant experiment taxes total private-education
spending. I also report an infeasible oracle tax that targets positional
outlays alone. The oracle avoids discouraging human-capital investment and
therefore gives taxation its best chance to reduce the race. Yet tax rates of
10 or 20 percent in either experiment change completed fertility by less than
0.002 children. When all families cut spending, the admission cutoff falls as
well, leaving relative positions largely unchanged, and the balanced-budget
rebate returns the revenue to households. Supplemental Appendix~B.5 shows why
the oracle tax avoids directly discouraging productive education. In general
equilibrium, however, cutoff adjustment makes both taxes nearly ineffective.

The source of the externality determines what an education tax can do. In the
comparison-motive model of \citet{MahlerTertiltYum2025}, a parental-investment
tax financing pronatal transfers implements the first best. KTY's tournament
variant provides a closer comparison. With tax revenue rebated lump-sum, tax
rates of 10 and 20 percent change fertility by only $-0.5$ and $-1.0$ percent.
KTY attribute this small response to two offsetting effects: the tax weakens the
externality but also raises the price of education. The near-null result here has
a different mechanism. The cutoff re-clears because the broad spending tax and
the oracle positional tax change the price of competing without changing
scarcity or the score-based assignment rule.

I next direct all revenue from a 20-percent education-spending tax to an equal
per-child transfer, setting $T=0$. The balanced-budget transfer is
$b/w=\vTaxTransferRate$ percent. The package raises completed fertility by
\vTaxTransferFert\ children and welfare by \vCETaxTransfer\ percent. Earmarking
is more pronatal than an equal rebate, but the effect remains small because the
assignment rule is unchanged.

\subsection{Welfare}
\label{sec:welfare}

Table~\ref{tab:welfare} reports consumption-equivalent changes in stationary
ex-ante adult welfare together with output per unit of labor, normalized to one
in the current benchmark. I use a baseline welfare measure that excludes the
child-time utility cost from both sides because its monetary value depends on
the unpriced common unit of the two education inputs. Allocations continue to
respond to that term, so the baseline measure is a partial-accounting
consumption equivalent evaluated at the model's allocation.
Supplemental Appendix~F defines the welfare criterion, decomposes this baseline
measure, reports the additional child-time effect and destination-value
sensitivity, and derives the transition law.

At a fixed admission share, random assignment raises welfare by \vCEexTime\
percent of consumption even though output per unit of labor falls by
\vLotteryProdCostBenchmark\ percent.
The decline comes from less ability-based sorting, while the positional spending
that disappears builds no human capital. The welfare gain is therefore net
of the productivity loss: the release of family resources and the reallocation
of destination value more than offset it.

Beyond the baseline measure, the model assigns a large benefit to releasing
children's time from positional preparation. This channel adds
\vChildTimeRelease\ utility units and accounts for \vChildTimeShare\ percent of
the full-model welfare increase. I report it separately rather than convert it
to consumption because the monetary value of the common education-input unit
is not estimated.

Two assumptions govern the baseline welfare gain. First, sufficient destination
value must survive random assignment. The gain falls to approximately zero if
the entire measured wage premium is match-specific and disappears under the
lottery. Second, the admission share must remain fixed as cohort size changes.
If the number of positions is fixed instead, additional births dilute admission
probabilities, and the welfare change is \vCEabsSlots\ percent even though
fertility rises by $0.199$.

The broad qualified lottery provides a more institutionally conservative
comparison. At $m=0.85$, it raises fertility by \vQualifiedFert\ and welfare
by \vCEqualified\ percent, with output per unit of labor equal to
\vQualifiedOutput\ relative to the benchmark. It captures three quarters of the
full lottery's fertility gain but only about half of its welfare gain. The
reason is that it combines the costs of both regimes: the broad pool weakens
productive sorting, while most families still enter the contest and continue
to spend for rank. Retaining an achievement threshold may make the reform
easier to implement, but does not by itself eliminate the social cost of the
race.

\begin{table}[htbp]
\centering\footnotesize
\caption{Welfare, fertility, and output under the policy experiments}
\label{tab:welfare}
\begin{threeparttable}
\begin{tabular}{lrrr}
\toprule
Policy & $\Delta$ fertility & CE (\%) & \shortstack{Output per\\unit of labor} \\
\midrule
Random assignment (lottery) & $+0.240$ & $+16.3$ & $0.993$ \\
Qualified lottery, $m=0.85$ & $+0.181$ & $+8.8$ & $0.988$ \\
Pronatal transfer, 2\% of the wage per child & $+0.012$ & $+1.0$ & $0.999$ \\
Pronatal transfer, 5\% of the wage per child & $+0.031$ & $+2.5$ & $0.997$ \\
Education-spending tax, 20\% & $-0.001$ & $+0.6$ & $1.002$ \\
Targeted positional tax, 20\% (oracle) & $-0.001$ & $+0.6$ & $1.003$ \\
Education tax, 20\%, with revenue-funded transfer & $+0.003$ & $+0.9$ & $1.002$ \\
Capacity $\times 1.10$, $\eta=0$ & $+0.010$ & $+39.7$ & $1.006$ \\
Capacity $\times 1.10$, $\eta=0.5$ & $+0.003$ & $+17.9$ & $1.005$ \\
Capacity $\times 1.10$, $\eta=1$ & $-0.004$ & $+0.4$ & $1.005$ \\
Capacity $\times 1.10$ with random assignment, $\eta=1$ & $+0.241$ & $+18.0$ & $0.998$ \\
\midrule
Random assignment, absolute number of positions fixed & $+0.199$ & $-19.8$ & $0.989$ \\
\bottomrule
\end{tabular}
\end{threeparttable}
\begin{fullwidthtablenotes}
CE is the consumption-equivalent change in stationary ex-ante adult
welfare, excluding the child-time term because the two education inputs
lack a common market conversion. Output per unit of labor is normalized
to one in the benchmark. Fiscal experiments balance the government budget.
The education-spending tax applies to productive and positional outlays.
The oracle tax applies only to the unobserved positional component.
The revenue-funded package sets the lump-sum adjustment to zero and
chooses the per-child transfer to exhaust education-tax revenue.
Except for the final row, each counterfactual holds its stated admission
share fixed across cohort sizes; the final row holds the absolute number
of positions fixed.
\end{fullwidthtablenotes}\end{table}

\FloatBarrier

The five-percent pronatal transfer raises welfare by \vCEtransferFive\ percent,
and both education-tax experiments change welfare by less than one percent in
absolute value. Capacity retains the ambiguity
seen in fertility: at $\eta=1$, a ten-percent expansion changes welfare by
\vCEcapEtaOne\ percent, while adding random assignment raises it by
\vCEcapLot\ percent. Among the transfer, tax, and capacity experiments, output per unit
of labor remains within 0.6 percent of the benchmark. Their welfare ranking is
therefore driven mainly by the resources released and the value of additional
destinations, not by large changes in aggregate productivity.

The welfare gains begin immediately and remain positive throughout the
transition. The first generation gains \vCEgenOne\ percent, after which average
ex-ante welfare converges to the steady-state gain. The gains are concentrated
among families at the bottom of the permanent-wage distribution, whose assignment prospects
improve under the lottery. The top quintile loses because score-based assignment
had allowed its resources to purchase a higher admission probability.

\section{Cohort Comparison}
\label{sec:cohorts}
\label{sec:cohort_decomposition}

The benchmark analysis explains a low level of fertility, not necessarily its
continued decline. To separate these margins, I compare adjacent birth cohorts
in KLIPS under a common construction. Completed fertility falls from 1.873
among women born in 1970--75 to 1.786 among women born in 1976--81. Both cohorts
are observed at ages 40--43, with completed fertility measured through age 43
and work hours constructed in the same way. The comparison asks whether the
contest intensified or desired family size changed.

The model provides a way to compare these explanations because it separates the
contest-and-technology block from preferences and household heterogeneity. In
the main specification, I recalibrate the younger cohort's parity tastes, weight
on children, fertility-choice dispersion, leisure weight, and ability process.
I hold the destination value, entry and effort costs, score technology, and
human-capital technology at their older-cohort values. Destination capacity,
the wage premium, the child time requirement, the private-education burden
gradient, and the KEEP restrictions also
remain common. I treat the younger cohort's lifecycle spending ratio and
KTY-style total-spending elasticity as descriptive because its children's education
histories are right-truncated. The elasticity is untargeted in both cohorts, and
the lifecycle ratio is omitted from the younger calibration; every retained
cohort target is constructed identically across cohorts. The specification asks whether
changes in preferences and heterogeneity can account for the fertility decline
with the education race held fixed. Supplemental Appendix~A documents the common
sample construction and the truncation of the younger cohort's education histories.

\begin{table}[htbp]
\centering\footnotesize
\caption{The two cohorts: contest frozen, preferences recalibrated}
\label{tab:cohort_fit}
\begin{threeparttable}
\begin{tabular}{lrr}
\toprule
 & 1970--75 & 1976--81 \\
\midrule
Cohort-target RMSE (\%) & 0.68 & 0.21 \\
Contest fertility & 1.8463 & 1.7799 \\
Lottery fertility & 2.0862 & 2.0545 \\
Lottery effect & 0.2399 & 0.2746 \\
Contest outlay / income & 0.0497 & 0.0477 \\
Positional share & 0.773 & 0.778 \\
Household participation & 0.808 & 0.770 \\
Destination value $B_E$ & 5.35 & 5.35 (frozen) \\
Fertility-choice dispersion $\sigma$ & 1.841 & 1.597 \\
Parity taste $\phi_1$ & 1.469 & 0.650 \\
Weight on children $\beta$ & 1.801 & 1.571 \\
\midrule
Target completed fertility & 1.8727 & 1.7857 \\
Change: model & \multicolumn{2}{r}{-0.0664} \\
Change: data & \multicolumn{2}{r}{-0.0871} \\
Reproduced & \multicolumn{2}{r}{76\%} \\
\bottomrule
\end{tabular}
\end{threeparttable}
\begin{fullwidthtablenotes}
Every external input is common across the two
columns. The younger column recalibrates only the
preference-and-heterogeneity block ($\phi_1$, $\phi_2$, $\phi_3$,
$\beta$, $\sigma$, $\nu$, $\sigma_\kappa$, $\rho_\kappa$) on the younger
cohort's moments; the contest and technology block (the destination value,
effort and entry costs, the score technology, and the human-capital
technology) is frozen at its 1970--75 value. Frozen rows repeat by design.
The displayed RMSE covers nine retained benchmark targets and eight
retained younger-cohort targets; the common burden restriction and KEEP
intervals are not included in that displayed statistic.
\end{fullwidthtablenotes}\end{table}

\FloatBarrier

Table~\ref{tab:cohort_fit} shows that preference and heterogeneity changes can.
With the contest held fixed, the
model produces completed fertility of 1.780 against a target of 1.786 and
reproduces \vCohortRepro\ percent of the decline between cohorts. The model also
fits the eight retained targets closely, with a root mean squared
relative error of \vKTYrmseYoung\ percent. At the same time, the economic size
of the contest barely changes: contest outlays remain
close to five percent of income, the positional share remains near 0.78, and a
lottery would raise younger-cohort fertility by \vEffectYoung. Education
competition therefore remains a large drag on fertility in the younger cohort
even though it does not account for the decline between cohorts.

Could a changing contest explain the same decline instead? Table~\ref{tab:cohort_hyp}
recalibrates alternative parameter blocks on the younger cohort's moments under
the same objective and evaluation budget. The first three rows free one block at
a time from the older-cohort calibrated values. The last row starts from the
preference-and-heterogeneity fit and then allows all sixteen parameters to move.
The comparison points away from a fiercer race. Freeing the equally sized
contest-and-technology block produces a root mean squared error of 8.99 percent,
more than forty times the 0.21 percent obtained by changing preferences and
heterogeneity. It also selects a smaller destination value and an overall weaker
contest rather than the intensification this explanation requires. Changing choice
dispersion alone fits no better and leaves too many three-child families. Even
when all sixteen parameters are subsequently allowed to move, the fit improves
only from 0.21 to 0.17 percent and every contest parameter remains within 2.8
percent of its older-cohort value. The additional adjustment occurs in the same
preference parameters selected by the restricted fit.

\renewcommand{\floatpagefraction}{0.70}
\renewcommand{\textfraction}{0.08}
\begin{table}[htbp]
\centering\footnotesize
\caption{Which block moves across cohorts? Restricted recalibrations}
\label{tab:cohort_hyp}
\begin{threeparttable}
\begin{tabular}{lrrr}
\toprule
Freed block & Free & \shortstack{Cohort-target\\RMSE (\%)} & \shortstack{Completed\\fertility} \\
\midrule
Preferences and heterogeneity (adopted) & 8 & 0.21 & 1.7799 \\
Contest and technology & 8 & 8.99 & 1.7827 \\
Choice dispersion $\sigma$ alone & 1 & 9.00 & 1.8076 \\
All parameters (extended) & 16 & 0.17 & 1.7798 \\
\bottomrule
\end{tabular}
\end{threeparttable}
\begin{fullwidthtablenotes}
Each of the first three rows starts from the 1970--75 benchmark, freezes every
other parameter at that cohort's value, and recalibrates only the named block on
the 1976--81 moments, under the same objective and evaluation cap. Target
completed fertility is 1.7857. The displayed RMSE is computed over the
eight retained targets; the common burden restriction and KEEP
intervals remain in the objective. These are descriptive recalibrations
with a common evaluation cap.
The first row is the younger-cohort specification used
throughout the paper. The final row is continued from the first row's
solution for a further 3{,}000 evaluations under the same objective; at
the final endpoint no contest parameter moves more than 2.8 percent
from the benchmark.
\end{fullwidthtablenotes}\end{table}

Having located the change in the preference-and-heterogeneity block,
Table~\ref{tab:pref_decomp} shows which parts of that block matter. It substitutes
the younger cohort's calibrated parameter values into the older economy one
group at a time.
The weight on children falls by thirteen percent and lowers completed fertility
by 0.043. Lower parity tastes contribute another 0.040; all three fall, with the
largest movement at two children, consistent with the decline in the two-child
share. These forces are partly offset by narrower choice dispersion, which
raises fertility by 0.023 when changed alone. Leisure and the ability process
contribute little. The separate changes sum to $-0.056$, compared with a joint
effect of $-0.066$, leaving a modest interaction of $-0.010$.

\begin{table}[htbp]
\centering\footnotesize
\caption{Decomposing the preference block, one group at a time}
\label{tab:pref_decomp}
\begin{threeparttable}
\begin{tabular}{lrr}
\toprule
Group & Change in fertility & Share \\
\midrule
Fertility-choice dispersion $\sigma$ & $+0.023$ & $-35$\% \\
Weight on children $\beta$ & $-0.043$ & $+64$\% \\
Parity taste levels $\Phi(n)$ & $-0.040$ & $+61$\% \\
Leisure and ability process & $+0.004$ & $-6$\% \\
\midrule
Sum of separate effects & $-0.056$ & \\
Actual joint change & $-0.066$ & \\
\bottomrule
\end{tabular}
\end{threeparttable}
\begin{fullwidthtablenotes}
Each row substitutes one group of the younger
cohort's preference parameters into the older economy, holding everything
else at the benchmark. Shares are relative to the actual joint change;
the interaction term is $-0.010$.
\end{fullwidthtablenotes}\end{table}

The cohort comparison therefore separates the low level of fertility from its
continued decline. The education race remains a large and persistent cost of
children in both cohorts. What changes is the valuation of family size,
especially at higher parities. This pattern is consistent with
\citet{FangLiu2026}'s cultural account: the continuing importance of educational
achievement appears in the stable contest, while the weakening pronatalist norm
appears in lower family-size preferences.
\FloatBarrier
\renewcommand{\floatpagefraction}{0.50}
\renewcommand{\textfraction}{0.20}

\section{Conclusion}
\label{sec:conclusion}

This paper develops an assignment-externality account of South Korea's education
race. Families share aspirations for a small set of coveted careers, and
educational effort affects who receives them. When all families spend more on
positional preparation, the cutoff rises without changing aggregate assignment,
so the race raises the equilibrium cost of a child. In the KEEP data,
private-education
spending predicts no measurable improvement in college-entrance scores,
preferred-career entry, or wages at age~29. Combined with observed spending in
the KLIPS data, these
estimates imply that most observed private-education spending is positional in
the calibrated model. Replacing score-based assignment with a
capacity-preserving lottery raises completed fertility from \vContest\ to
\vLottery, an increase of 13 percent. The released resources account for almost
the entire response. Output per unit of labor falls by only
\vLotteryProdCostBenchmark\ percent. With the admission share held fixed, the
partial-accounting welfare gain is \vCEexTime\ percent of consumption, with the
largest gains accruing to households lower in the permanent-wage distribution.

Other policies do much less. A five-percent pronatal transfer raises completed
fertility by \vTransferFive, and education taxes change it by almost nothing.
The cohort comparison separates the persistence of the contest from the decline
across generations.
Holding the contest at its older-cohort configuration and recalibrating
preferences and household heterogeneity reproduces \vCohortRepro\ percent of the
decline for women born in 1976--81. The same contest remains quantitatively
important in both cohorts. The model leaves housing and school-district choices,
as well as the processes that
change aspirations and family-size preferences, outside the analysis. Bringing
marriage and birth timing together with housing, labor-market institutions, and
preference formation is a natural direction for future research.

Institutional reform need not wait for aspirations to change.
The assignment externality arises jointly from shared aspirations for scarce
careers and a score-based rule that allocates them.
Examination achievement and elite-career aspirations have
deep roots across East Asia and are unlikely to change quickly
\citep{Sorensen1994,Seth2002,FangLiu2026}. The full-lottery experiment holds
these aspirations fixed and eliminates the model's assignment externality by
severing the link between positional effort and destination assignment. As long
as fine score differences continue to affect assignment, families retain an
incentive to compete. Qualified lotteries weaken the race only when the
qualified pool is broad. Nor does capacity expansion necessarily relieve
the pressure. It can attract new contestants and intensify competition. The
medical-school expansion episode illustrates how quickly demand can pivot toward
a widened gate. Near-term reform should therefore target the score--assignment
link rather than capacity alone.

A second implication concerns the distinction between persistently low
fertility and its further decline. The same contest depresses fertility in both
cohorts, but a lower utility weight on children and lower parity tastes explain
most of the decline between them. These family-size preferences are distinct
from the shared career aspirations that generate the assignment externality.
Longer-run reform must therefore work on two margins. Stronger vocational
pathways, as advocated by \citet{FangLiu2026}, can diversify preferred
destinations. Better pay and social standing for technical and production
careers could make those pathways credible. Workplace, housing, and family
institutions that support marriage and parenthood may also gradually reshape
desired family size. Educational reform can lower the institutional cost of
children. Sustained fertility recovery will also depend on whether family life
becomes more attractive and feasible for younger generations.

\clearpage
\thispagestyle{empty}
\null
\vspace{0.12\textheight}
\begin{center}
{\LARGE\bfseries Supplemental Appendix for\par}
\vspace{1.5em}
{\Large ``The Race for Elite Destinations:\\
Education Competition and Low Fertility in Korea''\par}
\vspace{3em}
{\large Dongwoo Kim\par}
\vspace{1.5em}
{\large August 2026\par}
\end{center}
\vfill
\clearpage

\appendix
\renewcommand{\thesection}{\Alph{section}}
\renewcommand{\thesubsection}{\thesection.\arabic{subsection}}
\renewcommand{\thesubsubsection}{\thesubsection.\arabic{subsubsection}}
\numberwithin{equation}{section}
\numberwithin{table}{section}
\numberwithin{figure}{section}
\setcounter{footnote}{0}
\setcounter{proposition}{0}
\setcounter{lemma}{0}
\setcounter{corollary}{0}
\setcounter{definition}{0}
\renewcommand{\theproposition}{\thesection.\arabic{proposition}}
\renewcommand{\thelemma}{\thesection.\arabic{lemma}}
\renewcommand{\thecorollary}{\thesection.\arabic{corollary}}
\renewcommand{\thedefinition}{\thesection.\arabic{definition}}

\section{Data Construction}
\label{app:data}

This appendix documents the samples and variables behind the household moments,
education spending and time inputs, preferred-destination measures, and
private-education returns for the calibration of the quantitative model.
The data come from the Korean Labor and Income Panel
Study (KLIPS), the Private Education Expenditures Survey (PEES), the Korean Time Use
Survey (KTUS), the Korean Education and Employment Panel (KEEP), and official
workplace-preference and youth labor-force statistics.

\subsection{KLIPS Household Sample}
\label{app:klips_sample}

KLIPS has followed households and individuals annually since 1998. I use all
waves through wave 27 (2024) and update the sample construction in
\citet{KTY2024} (hereafter KTY). I retain married or cohabiting couple
households in which the woman was born in 1970--75 and appears at ages 40--43
in at least three waves, requiring valid observations for all calibration
variables. From the common set of valid observations, I construct household
income, labor supply, fertility, and education expenditure, average income over
the age window, and count children observed by age 43.
The resulting benchmark sample contains 770 households. Applying the same rules
to women born from 1976 through 1981 yields 1,129 households. Both cohorts span
the full age window in wave 27; women born in 1981 turn 43 in 2024.
Completed fertility is measured through age 43 in both cohorts. KTY
compute completed fertility before imposing the age filter, which gives the older
cohort more post-43 observations once later waves are added. I instead apply the
same age cap to each cohort, as fertility after 43 is small in both.
Education spending requires a different sample. Following KTY, I use the full
birth cohort whenever a resident child is aged 0--24, without imposing the
mother-age window or the three-observation rule. These moments trace the child's
education life cycle; restricting mothers to ages 40--43 would retain only one
part of it. Each sample rule is applied identically across cohorts.

\subsection{Fertility and Income}
\label{app:fertility_income_data}

Completed fertility counts children observed by the wife's 43rd birthday. I
also top-code families with four or more children at three to match the model's
fertility support, $\{0,1,2,3\}$.
The headline figures 1.873 and 1.786 are top-coded means, $p_1+2p_2+3p_{3+}$;
the uncoded means by age 43 are 1.890 and 1.798. The parity probabilities sum to one:
\begin{equation}
\widehat p_0+\widehat p_1+\widehat p_2+\widehat p_{3+}=1.
\end{equation}
The calibration targets $p_1$, $p_2$, and $p_{3+}$; childlessness is the residual
at the aggregate level. Bottom-quintile childlessness is targeted separately.

Permanent household income is the average CPI-adjusted household income in the
restricted age window. It includes labor and capital income of both partners and excludes
temporary transfers. Income quintiles are formed within each cohort.
The Gini coefficient is computed on this permanent-income measure.

Following KTY, I estimate the fertility--income elasticity from
permanent-income-quintile means rather than household observations. This
aggregation retains childless households, for whom log fertility is undefined.
Households are sorted into five equal-mass quintiles by permanent income; mean
fertility $\overline n_q$ and mean income $\overline y_q$ are formed within each
quintile; and the elasticity is the slope of
\begin{equation}
\log\overline n_q=a+\varepsilon_n\log\overline y_q+u_q,
\qquad q=1,\ldots,5 .
\label{eq:data_fertility_elasticity}
\end{equation}
Both $\overline n_q$ and the elasticity use the three-child top code, as does the
model analogue formed from the stationary distribution. The resulting targets
are 0.079 for the benchmark cohort and 0.068 for the younger cohort. Without the
top code they are 0.084 and 0.063. I retain the uncapped construction for the KTY
replication below because KTY do not top-code fertility. Fertility rises with
income in both cohorts. In the benchmark profile reported in
Table~\ref{tab:app_quintile_data}, work hours also rise with income, while
childlessness falls sharply.

\begin{table}[htbp]
\centering
\small
\caption{Benchmark KLIPS profiles by permanent-income quintile}
\label{tab:app_quintile_data}
\begin{threeparttable}
\begin{tabular}{lrrrrrr}
\toprule
Quintile & Fertility & Childless & One child & Two children & Three+ & Work hours \\
\midrule
1 & 1.773 & 0.0649 & 0.1883 & 0.6558 & 0.0909 & 0.2446 \\
2 & 1.857 & 0.0455 & 0.2078 & 0.5909 & 0.1558 & 0.2685 \\
3 & 1.864 & 0.0195 & 0.2273 & 0.6234 & 0.1299 & 0.3005 \\
4 & 1.877 & 0.0195 & 0.2078 & 0.6494 & 0.1234 & 0.3101 \\
5 & 1.994 & 0.0130 & 0.1818 & 0.6039 & 0.2013 & 0.3356 \\
\bottomrule
\end{tabular}
\end{threeparttable}
\begin{fullwidthtablenotes}
Work hours are spouse-average weekly hours divided by 100. Fertility is
top-coded at three children, as in the targeted elasticity. Quintile profiles
are not individually targeted.
\end{fullwidthtablenotes}
\end{table}

\subsection{Labor Supply}
\label{app:labor_data}

Weekly hours are regular main-job hours for wage workers and
average main-job hours for non-wage workers. Missing spouse hours are set to
zero. Couple labor supply is the average of the two spouses' weekly hours
divided by 100:
\begin{equation}
\widehat\ell_i
=\frac{\text{hours}_{i,m}+\text{hours}_{i,f}}{2\times100}.
\label{eq:data_hours}
\end{equation}
This normalization puts hours directly in the model's time-endowment units. The
targets are 0.292 for the benchmark cohort and 0.278 for the younger cohort.

\subsection{Education Expenditure}
\label{app:education_data}

KLIPS reports total education expenditure at the household level and
private-education expenditure at the child level. The lifecycle measure of total
education expenditure follows KTY's construction and restricts the sample to
households with one co-resident child, so the household total can be assigned to
that child. For each child age from birth to 24, mean total education expenditure
$\widehat X_a^{\mathrm{tot}}$ and mean after-tax household income $\widehat Y_a$
are computed among households whose only co-resident child is that age. I sum
the total-expenditure and income profiles separately and divide the former by
the latter:
\begin{equation}
\widehat m_x^{\mathrm{tot}}
=
\frac{\sum_{a=0}^{24}\widehat X_a^{\mathrm{tot}}}
{\sum_{a=0}^{24}\widehat Y_a}
=0.0883.
\label{eq:lifecycle_education_ratio}
\end{equation}
Taking the ratio after summation avoids giving excessive weight to ages with low
income. Private education accounts for 72.8 percent of total education
expenditure, implying a lifecycle private-education ratio of 0.0643. Because
public education is common across households, the model's household-specific
education inputs correspond to private-education spending. The calibration therefore
targets the private-education ratio.

KTY's investment--income elasticity uses total education expenditure and
imposes no one-child restriction.
Children are grouped into four schooling stages (ages 0--6, 7--12, 13--15,
16--18). Within each stage, mean total education expenditure per child and mean
income are computed by income quintile; the stage means are then averaged with
weights equal to the number of years in each stage (7, 6, 3, 3), and the
elasticity is the slope of log mean expenditure on log mean income across the
five quintiles. This construction reproduces KTY's published figure to within
0.5 percent and gives 0.659 in the benchmark cohort. It is not a calibration
target because it is defined over total education expenditure, whereas the
model includes only household-varying private outlays.

The younger cohort does not yet provide a representative child education life
cycle, so I report its education moments as descriptive rather than calibration
targets. The lifecycle spending ratio is right-truncated: the one-child cells at
ages 22, 23, and 24 contain only 17, 6, and 3 observations, compared with 146,
136, and 108 in the benchmark cohort. The KTY-style elasticity ends at high
school, but for the younger cohort the available secondary-school observations
are disproportionately drawn from women who gave birth relatively early. In the
high-school stage, 73 percent of observations correspond to births before
maternal age 29, compared with 47 percent in the benchmark cohort, and mean
maternal age at birth is 26.7 rather than 28.7.

Table~\ref{tab:app_stage_elasticity} shows that the younger cohort's lower
KTY-style elasticity is concentrated in middle and high school, where
its sample is thinner and selected toward early births. Across preschool and
elementary school, where both cohorts are well represented, the difference is
only 0.007, compared with 0.032 across all stages. I therefore treat the younger
cohort's education ratio of 0.082 and elasticity of 0.627 as descriptive; later
waves are needed for a full lifecycle comparison.

\begin{table}[htbp]
\centering
\small
\caption{KTY-style investment--income elasticity by schooling stage}
\label{tab:app_stage_elasticity}
\begin{threeparttable}
\begin{tabular}{lrrrrr}
\toprule
Stage & 1970--75 & Obs. & 1976--81 & Obs. & Change \\
\midrule
Preschool (0--6) & 0.553 & 5,849 & 0.545 & 7,787 & $-0.007$ \\
Elementary (7--12) & 0.614 & 7,062 & 0.610 & 7,481 & $-0.004$ \\
Middle (13--15) & 0.698 & 4,279 & 0.632 & 3,236 & $-0.065$ \\
High (16--18) & 0.861 & 3,950 & 0.761 & 1,910 & $-0.099$ \\
\midrule
All stages (7/6/3/3) & 0.659 & & 0.627 & & $-0.032$ \\
Preschool and elementary & 0.583 & & 0.576 & & $-0.007$ \\
\bottomrule
\end{tabular}
\end{threeparttable}
\begin{fullwidthtablenotes}
Elasticities are computed by the construction described above, applied one
stage at a time. Observation counts are child-year observations with non-missing
expenditure. Neither column enters the calibration objective; the
table reproduces the KTY-style construction and shows where the younger cohort's
truncation becomes consequential. Changes are computed from unrounded values.
\end{fullwidthtablenotes}
\end{table}

Table~\ref{tab:app_kty_replication} applies the same code to KTY's data under
their sample rules and reproduces all published moments within seven-tenths of
a percent. Their completed-fertility measure is the
whole-panel maximum taken before the age filter, and their work-hours measure
comes from a separate time-use file that pools birth cohorts and spouses aged
26--50. The last two columns apply this paper's common cohort rules to wave-27
data and report the resulting harmonized measurements.

\begin{table}[htbp]
\centering
\small
\caption{Replication of \citet{KTY2024} and wave-27 measurements}
\label{tab:app_kty_replication}
\begin{threeparttable}
\begin{tabular}{lrrrrr}
\toprule
Moment & KTY & Reproduced & Error & 1970--75 & 1976--81 \\
\midrule
One-child share & 0.196 & 0.196 & $-0.1$\% & 0.203 & 0.220 \\
Two-child share & 0.631 & 0.631 & $-0.0$\% & 0.625 & 0.568 \\
Three-or-more share & 0.144 & 0.144 & $+0.1$\% & 0.140 & 0.143 \\
Income Gini & 0.263 & 0.263 & $-0.2$\% & 0.263 & 0.241 \\
Work hours & 0.301 & 0.301 & $+0.0$\% & 0.292 & 0.278 \\
Fertility--income elasticity & 0.082 & 0.082 & $-0.6$\% & 0.079 & 0.068 \\
Bottom-quintile childlessness & 0.053 & 0.053 & $-0.7$\% & 0.065 & 0.097 \\
Education spending/income & 0.092 & 0.092 & $-0.3$\% & 0.088 & [0.082] \\
Investment--income elasticity & 0.698 & 0.695 & $-0.5$\% & 0.659 & [0.627] \\
\midrule
Completed fertility & 1.890 & 1.890 & $+0.0$\% & 1.873 & 1.786 \\
\bottomrule
\end{tabular}
\end{threeparttable}
\begin{fullwidthtablenotes}
``Reproduced'' applies the code of this paper to the data shipped with
\citet{KTY2024} under their own sample rules, including the uncapped completed
fertility of their cohort file and the separate 26--50 sample of their time-use
file. The last two columns apply the same code to KLIPS through wave 27 under
this paper's conventions, so the hours row compares two different samples by
construction. Bracketed entries are additionally right-truncated. Intergenerational
persistence has no cohort analogue in these data and is held at its external
value in both cohorts.
\end{fullwidthtablenotes}
\end{table}

\subsection{Income Distribution and Intergenerational Persistence}
\label{app:income_moments}

The weighted Lorenz curve for permanent household income yields Gini coefficients
of 0.263 in the benchmark cohort and 0.241 in the younger cohort. KLIPS does not
provide comparable cohort-specific estimates of intergenerational income persistence,
so I use KTY's target of 0.320 in both calibrations. The model counterpart is the
elasticity of children's realized adult earnings with respect to parental permanent
earnings, with the destination wage premium included in children's earnings.

\subsection{Preferred-Destination Measurement}
\label{app:destination_measurement}

\begin{table}[htbp]
\centering
\small
\caption{Operational definition of preferred destinations}
\label{tab:app_destination_definition}
\begin{threeparttable}
\begin{tabular}{p{0.27\textwidth}p{0.64\textwidth}}
\toprule
Component & KLIPS definition \\
\midrule
Licensed professions &
Employed or self-employed in KSCO three-digit codes 241, 242, 261, 271, or 272.
These codes cover selected medical, legal, accounting, and related professions. \\
Teachers &
Employed in KSCO codes 251 or 252, covering university and school teachers.
Academy instructors are excluded. \\
Government &
Regular wage worker at a government institution with a three-digit occupation
code below 400. Civil-service grade is unavailable. \\
Public institutions &
Regular wage worker at a public institution or public enterprise. \\
Large or foreign firms &
Regular wage worker at a private or foreign workplace with at least 300
employees. \\
\bottomrule
\end{tabular}
\end{threeparttable}
\begin{fullwidthtablenotes}
Occupations use the sixth revision of the Korean Standard
Classification of Occupations. Workplace and firm-size variables are from the
main-job questionnaire.
\end{fullwidthtablenotes}
\end{table}

Using the KLIPS wave-27 integrated person file, I classify an adult as being in a
preferred destination if the person falls into any of the five categories in
Table~\ref{tab:app_destination_definition}. If more than one category applies,
I assign the person to the first applicable category in the order shown and
count the person only once. Using KLIPS cross-sectional weights, I measure capacity as the share of adults
ages 30--45 who fall into at least one preferred-destination category. The
denominator includes wage workers, self-employed workers, and nonemployed adults:
\begin{equation}
\widehat q_E
=
\frac{\sum_i\omega_i\1\{D_i=1\}}
{\sum_i\omega_i}
=0.1845.
\label{eq:capacity_measure}
\end{equation}
The current-wage premium is estimated among wage workers ages 25--54 with positive
pay and hours:
\begin{equation}
\log w_i^{\mathrm{hour}}
=
\delta_E D_i
+\sum_{k=1}^{4}\alpha_k\text{age}_i^k
+\gamma\,\text{female}_i
+\sum_j\psi_j\1\{\text{education}_i=j\}
+u_i.
\label{eq:wage_premium_regression}
\end{equation}
The regression uses KLIPS cross-sectional weights and heteroskedasticity-robust
standard errors. The estimate is
\begin{equation}
\widehat\delta_E=0.25739,
\qquad \operatorname{se}(\widehat\delta_E)=0.01416.
\end{equation}
The 95 percent confidence interval is $[0.2296,0.2851]$.

As a mover-based check for the large-firm component, I use KLIPS waves 22--27
and retain workers who are ages 25--54 when first observed. I classify each
worker's first observed job change. The treated group moves from a wage
job outside the category to a regular job at a firm with at least 300 employees;
the comparison group moves between wage jobs outside the category. The sample
contains 61 large-firm movers and 1{,}203 comparison movers. Let $T_i$ denote
the move wave, let $M_i$ indicate a large-firm mover,
and define
$\mathcal K=\{-4,-3,-2,0,1,2,3,4\}$, with $k=-1$ the omitted reference period. I estimate
\begin{equation}
\log w_{it}^{\mathrm{hour}}
=
\mu_i+\lambda_t
+\sum_{k\in\mathcal K}\tau_k\1\{t-T_i=k\}
+\sum_{k\in\mathcal K}\beta_k M_i\1\{t-T_i=k\}
+u_{it},
\label{eq:mover_wage_regression}
\end{equation}
where $\mu_i$ and $\lambda_t$ are person and wave fixed effects and standard
errors are clustered by person. The average post-move interaction is
$\frac{1}{5}\sum_{k=0}^{4}\widehat\beta_k=0.0759$, with a standard error of
0.0392 and a 95 percent confidence interval of $[-0.0010,0.1527]$.

\subsection{Time Inputs from the Korean Time Use Survey}
\label{app:ktus}

The time-use statistics in Section~\ref{sec:fertility_education} come from
KTUS microdata spanning 2004--2024. I convert diary minutes to weekly hours by
weighting weekday, Saturday, and Sunday observations by five, one, and one and
applying the survey's diary weights. Out-of-school study includes private-academy
instruction, self-study outside school, and other out-of-school study;
study-related travel uses the survey's learning-related travel category.
Household income bands come from the household-information file. I map the 2004
activity codes into the same categories used in later waves. KTUS records how
children spend their time but not whether study builds skills or improves their
rank in the admission contest, so the model determines that split.

\subsection{KEEP Sample Accounting}
\label{app:keep_sample}

Of the 2{,}000 students in the 2004 KEEP cohort, 574 report positive earnings in
2015 and enter the wage regression of Section~\ref{sec:keep_returns}. On
the regression's baseline covariates,
they closely resemble the excluded students: the differences are $0.01$
pooled standard deviations in log household income, $0.07$ in log
private-education spending, and $0.13$ in teacher-assessed achievement.
Selection into the College Scholastic Ability Test (CSAT) sample is more
pronounced. CSAT takers report log private-education spending $0.37$ standard
deviations above non-takers, so the score estimates describe a more heavily
investing subset of the cohort.

\subsection{Private-Education Participation and the Burden Gradient}
\label{app:participation}

Participation is measured as positive private-education spending in a
household-year. The sample contains 1{,}579 households in the 1970--75 cohort
and 18{,}136 household-year observations with a resident child aged 24 or
younger and positive income. No maternal-age restriction is imposed. Income
quintiles are based on each household's mean income over these observations.
Participation rises from 0.572 in the bottom income quintile to 0.850 in the
top. The full profile is 0.572, 0.722, 0.791, 0.835, and 0.850 across quintiles.
The sample participation rate is 0.766, and the top-to-bottom ratio is 1.486. These
statistics are untargeted because positive private-education spending does not
distinguish productive education from contest entry. Across the same
quintiles, the mean ratios of private-education spending to income are 0.130, 0.110, 0.112,
0.121, and 0.100. Their bottom-to-top ratio of 1.297 is the targeted burden
gradient.

\section{Proofs and Extensions of the Analytical Model}
\label{app:proofs}

This appendix develops the analytical results used in Sections~\ref{sec:analytical}
and~\ref{sec:education_taxes}. Unless stated otherwise, it maintains the
analytical environment of Section~\ref{sec:analytical}. The results
formalize private effort, the planner's allocation, the capacity limits, and the
family-size wedge. The final subsection introduces a partial-equilibrium
extension comparing a broad education-spending tax with an oracle tax on
positional outlays.

\subsection{Environment}
\label{app:proof_environment}

A family with state $(a,y)$ and parity $n\geq1$ chooses productive education $x$
and positional effort $e$ per child, with consumption $c=y-nx-nC(e)$,
$C(0)=C'(0)=0$, and $C'(e),C''(e)>0$ for $e>0$. The child's score is
$z=a+\theta\log h(x)+g(e)+\varepsilon$ with $h'>0$, $h''<0$, $\theta>0$,
$g'>0$, and $\varepsilon\sim F_\varepsilon$ with density
$f_\varepsilon$. A child is admitted when $z\geq v$, so
\begin{equation}
p=1-F_\varepsilon(\zeta),
\qquad
\zeta(x,e)\equiv v-a-\theta\log h(x)-g(e).
\label{app-eq:admission}
\end{equation}
The destination is worth $\Delta=\delta+B>0$ to the parent. Conditional on
parity, the objective is
$u(c)+\Phi(n)+\beta n\{u_c(h(x))+\Delta p\}$ with $u'>0$, $u''<0$, $u_c'>0$,
$u_c''\leq0$. Differentiating in $e$ and dividing by $n$ gives the interior
condition
\begin{equation}
u'(c)\,C'(e)
=
\beta\Delta\,f_\varepsilon(\zeta(x,e))\,g'(e).
\label{app-eq:effort-foc}
\end{equation}
The cutoff clears a fixed capacity: the birth-weighted share of admitted
children equals $q\in(0,1)$.
In a symmetric equilibrium, capacity clearing requires
\begin{equation}
v=a+\theta\log h(x)+g(e)+F_\varepsilon^{-1}(1-q).
\label{app-eq:symmetric_cutoff}
\end{equation}

\subsection{Private Effort and the Planner}
\label{app:proof_allpay}

\begin{lemma}[All-pay effort]
\label{lem:app_allpay}
Suppose $C'(0)=0$, $g'(0)>0$, and
$f_\varepsilon(\zeta(x,0))>0$. If $\Delta>0$, any optimal effort choice satisfies
$e>0$. Under capacity-preserving lottery assignment the unique effort choice is
$e=0$, while productive education may remain positive.
\end{lemma}

\emph{Proof.} At zero effort, the marginal resource cost vanishes and
\begin{equation}
\left.\frac{\partial V}{\partial e}\right|_{e=0}
=\beta n\Delta f_\varepsilon(\zeta(x,0))g'(0)>0.
\end{equation}
Zero effort therefore cannot be optimal. Any optimum has $e>0$.
Under the lottery, admission is $q$ at every effort level, so
$\partial V^{L}/\partial e=-nu'(c)C'(e)<0$ for $e>0$ and the unique optimum is
$e^{L}=0$. Productive education retains the term
$\beta n\,u_c'(h)h'(x)$, which does not depend on the assignment rule, so $x$ need
not fall to zero. \qed

The assumptions say that, holding the cutoff fixed, a unilateral increase in
effort from zero raises the family's admission probability. Since $C'(0)=0$,
the marginal cost of effort is zero at $e=0$. The admission response is zero
under lottery assignment because the score no longer affects assignment. A common
increase in effort is instead absorbed by the cutoff as in
Lemma~\ref{lem:absorption}.

\begin{corollary}[Dissipation of the common component]
\label{cor:planner}
A planner who preserves destination assignments and values consumption, human
capital, and destination outcomes sets the common positional component of effort
to zero.
\end{corollary}

\emph{Proof.} Once the cutoff adjusts, the cutoff-absorption result in
Lemma~\ref{lem:absorption} of the main text implies that the common
component of $e$ has no assignment benefit. Removing it raises
consumption by $nC(e)$ without changing human capital or assignments, and hence
raises welfare whenever $e>0$. The result does not extend to productive
education because reducing $x$ lowers every child's human capital. \qed

\subsection{Capacity Limits in the Analytical and Quantitative Models}
\label{app:proof_value_scarcity}

\begin{lemma}[Capacity limits]
\label{lem:app_value_scarcity}
Under the conditions of Lemma~\ref{lem:app_allpay}, suppose $\Delta>0$ and
$f_\varepsilon(F_\varepsilon^{-1}(1-q))\to0$ as $q$ approaches zero or one. Then
equilibrium effort and the associated resource distortion vanish at both capacity
limits and can attain a maximum at an interior $q$.
\end{lemma}

\emph{Proof.} Substituting equation~\eqref{app-eq:symmetric_cutoff} into
equation~\eqref{app-eq:effort-foc}, the interior symmetric condition is
\begin{equation}
u'(c)C'(e)
=\beta\Delta\,f_\varepsilon\!\left(F_\varepsilon^{-1}(1-q)\right)g'(e).
\label{app-eq:symmetric_effort_foc}
\end{equation}
As $q$ approaches either limit, the right-hand side converges to zero by
assumption. Since $c\leq y$ and $u$ is concave, $u'(c)\geq u'(y)>0$.
Continuity with $C'(e)>0$ for $e>0$ and $g'(e)<\infty$ therefore implies
$e\to0$.
Effort is positive at interior capacities with positive density, so both effort
and its resource cost can peak at an interior $q$. The location of that peak
depends on preferences, costs, and the score distribution. \qed

Lemma~\ref{lem:app_value_scarcity} uses the additive score of the analytical
model. The quantitative model's normalized logistic rule changes the lower
capacity boundary.

\begin{corollary}[Vanishing private return at the capacity boundaries]
\label{cor:qlimit}
Let admission follow the quantitative model's rule
$P_E(S;v)=\Lambda\!\left((S-v)/(\sigma_s v)\right)$, with $\Lambda$ the logistic
function, $S>0$, and $v>0$. The symmetric rule can implement shares
$q\in(q_{\min},1)$, where $q_{\min}=\Lambda(-1/\sigma_s)$. Holding the cutoff
fixed, the private marginal return of the score at a symmetric point is
\begin{equation}
\frac{\partial P_E}{\partial S}
=\frac{q(1-q)\left[1+\sigma_s\operatorname{logit}(q)\right]}{\sigma_s S},
\label{app-eq:qlimit}
\end{equation}
which converges to zero as $q\to1$ and as $q\downarrow q_{\min}$.
\end{corollary}

\emph{Proof.} At the symmetric point, $\Lambda((S-v)/(\sigma_s v))=q$ gives
$(S-v)/(\sigma_s v)=\operatorname{logit}(q)$ and hence
$v=S/[1+\sigma_s\operatorname{logit}(q)]$. Treating $v$ as fixed,
$\partial P_E/\partial S=\Lambda'\cdot(\sigma_s v)^{-1}$ with
$\Lambda'=q(1-q)$, which on substituting for $v$ yields
equation~\eqref{app-eq:qlimit}. As $q\to1$, $q(1-q)\to0$ while
$(1-q)\operatorname{logit}(q)\to0$ as well, so the numerator vanishes; the
derivative therefore converges to zero. At the lower boundary,
$1+\sigma_s\operatorname{logit}(q)\to0$ while $q(1-q)$ remains finite, so the
derivative again vanishes. Shares below $q_{\min}$ are not attainable with
positive $S$ and $v$ under this normalization: as $v\to\infty$, the logistic
argument approaches $-1/\sigma_s$. \qed

\subsection{The Family-Size Wedge}
\label{app:proof_fertility_wedge}

Section~\ref{sec:analytical_fertility} describes the dilution channel in
prose. Here I state it formally. For a family state $a$ at parity $n$,
write $W_n(a)\equiv n\,p_n(a;v)$ for the expected number of children who receive
the preferred destination.

\begin{proposition}[Contest-induced fertility wedge]
\label{prop:fertility_wedge}
Evaluated at the optimized choices for each parity, the destination-value
contribution to the discrete gain from $n$ to $n+1$ children is
$\beta\Delta[W_{n+1}(a)-W_n(a)]$ under the contest and $\beta\Delta q$ under the
capacity-preserving lottery. The contest lowers the marginal value of child $n+1$
whenever $W_{n+1}(a)-W_n(a)<q$. Moreover, $W_{n+1}(a)<W_n(a)$ if
$p_{n+1}(a;v)/p_n(a;v)<n/(n+1)$.
\end{proposition}

\emph{Proof.} The destination-value part of optimized utility is
$\beta\Delta W_n(a)$. Its contribution to the discrete utility change from $n$ to
$n+1$ is $\beta\Delta\{W_{n+1}(a)-W_n(a)\}$. Under the lottery each child wins
with probability $q$, so the corresponding term is
$\beta\Delta\{(n+1)q-nq\}=\beta\Delta q$. The contest contribution is smaller
precisely when $W_{n+1}(a)-W_n(a)<q$, proving the first statement. For the
second,
\begin{align*}
W_{n+1}(a)<W_n(a)
&\Longleftrightarrow (n+1)p_{n+1}(a;v)<np_n(a;v)\\
&\Longleftrightarrow
\frac{p_{n+1}(a;v)}{p_n(a;v)}<\frac{n}{n+1},
\end{align*}
where the last equivalence uses $p_n(a;v)>0$. \qed

\subsection{Broad and Oracle-Targeted Education Taxes}
\label{app:tax_derivation}

Observed private-education spending combines productive education and positional
outlays. Let $\tau$ denote a broad tax on both components and let $\tau^{P}$
denote an infeasible oracle tax on positional outlays alone. To compare their
direct effects, hold family size and the admission cutoff fixed, condition on
contest participation, and let utility be linear in consumption. Omitting terms
that do not vary with $x$ or $e$, the family chooses
\begin{align}
\max_{x,e\geq0}\quad
y-n(1+\tau)x-n(1+\tau+\tau^{P})C(e)
+\beta n\{u_c(h(x))+\Delta p(x,e)\}.
\label{app-eq:tax_objective}
\end{align}
The interior first-order conditions are
\begin{align}
1+\tau
&=\beta\{u_c'(h)h'(x)+\Delta p_x(x,e)\},
\label{app-eq:broad_tax_x_foc}\\
(1+\tau+\tau^{P})C'(e)
&=\beta\Delta p_e(x,e).
\label{app-eq:targeted_tax_e_foc}
\end{align}
Setting $\tau^{P}=0$ gives the broad spending tax, while setting $\tau=0$
gives the oracle tax. The oracle raises the price of an input whose only return
is assignment without directly taxing productive education. Productive
education continues to raise human capital and child value, as reflected in
$u_c'(h)h'(x)$. The oracle therefore isolates the strongest case
for targeted taxation. Section~\ref{sec:education_taxes} solves the full
rebated general equilibrium, including endogenous contest entry and the fixed
fee. Once the cutoff re-clears, both experiments have near-zero fertility
effects. It also reports a revenue-funded package that fixes the broad tax at
20 percent, sets the lump-sum adjustment to zero, and uses all tax revenue for
an equal transfer per child.

\section{Quantitative Model: Additional Detail}
\label{app:model_details}

This appendix completes the quantitative model in Section~\ref{sec:model}.
It states the conditional household problem, constructs the intergenerational
transition, and describes the equilibrium algorithm used in the calibration.

\subsection{Conditional Household Problem}
\label{app:conditional_household}

For a state $s=(h,\kappa_p)$, let
$\Pi_\kappa(\kappa'|\kappa_p)$ denote the probability of child ability state
$\kappa'$ conditional on parental family-ability state $\kappa_p$. Given a
proposed pair $(v,T)$, the household solves a continuous problem for every
$n\in\{1,2,3\}$ and participation state $r\in\{0,1\}$:
\begin{align}
V_{nr}(s;v,T)
=\max_{x,e,\ell}\quad&
\log\!\left(\frac{1.5c_{nr}}{1.5+0.3n}\right)
+\nu\frac{(1-\ell-n\lambda)^{1-\gamma}-1}{1-\gamma}
+\Phi(n) \notag\\
&+\beta n\sum_{\kappa'}\Pi_\kappa(\kappa'|\kappa_p)
\Bigl[
u_c\!\bigl(y_0'(x,\kappa')\bigr)
+P_E\!\bigl(S(h'(x,\kappa'),e);v\bigr)\,\Delta u(x,\kappa')
\Bigr]\notag\\
&\quad- n\,\psi\,\frac{(x+e)^{1+\varphi}}{1+\varphi}
\label{eq:app_household_problem}\\
\text{subject to}\quad&
c_{nr}=wh\ell+T-n\bigl[x+r\bigl(F+C(e)\bigr)\bigr]>0,\notag\\
&0\leq\ell<1-n\lambda,\quad x\geq0,\notag\\
&e=0\ \text{if }r=0,\qquad e\geq0\ \text{if }r=1.
\notag
\end{align}
Here $y_0'(x,\kappa')=wh'(x,\kappa')$ is child earnings outside the preferred
destination, and
\[
\Delta u(x,\kappa')
=u_c\!\bigl(y_0'(x,\kappa')\exp(\delta_E)\bigr)
-u_c\!\bigl(y_0'(x,\kappa')\bigr)+B_E
\]
is the utility gain from assignment. Participation determines access to
positional effort, not assignment eligibility. A non-entrant sets $e=0$ but
still competes through productive human capital.

Let $\pi_{r|n}(s)$ denote the probability of participation state $r$ conditional
on parity. Type-I extreme-value participation shocks with scale $\sigma_s$ give
\begin{align}
\pi_{r|n}(s)
&=
\frac{\exp\{V_{nr}(s;v,T)/\sigma_s\}}
{\sum_{j=0}^{1}\exp\{V_{nj}(s;v,T)/\sigma_s\}},
\label{eq:app_entry_logit}\\
\overline U_n(s;v,T)
&=\sigma_s\log\!\left[
\sum_{r=0}^{1}\exp\{V_{nr}(s;v,T)/\sigma_s\}
\right].
\label{eq:app_entry_inclusive}
\end{align}
For $n=0$, only $r=0$ is available, education and effort are zero, and the
household chooses labor supply. Thus $\overline U_0=V_{00}$ and
$\pi_{0|0}=1$. The parity logit in equation~\eqref{eq:fertility_logit}
uses these inclusive values. Define
$\pi_{nr}(s)=\Pr(n|s)\pi_{r|n}(s)$ as the joint probability of parity and
participation.

Child income utility is $u_c(y)=(y^{1-\eta_c}-1)/(1-\eta_c)$ with
$\eta_c=2$; the logarithmic form in Section~\ref{sec:analytical} is its
$\eta_c\to1$ limit. Under CRRA utility, the wage component of $\Delta u$ varies
with child income, while $B_E$ is additive. Only $\delta_E$ enters measured
earnings and the next-generation income state.

The labor choice is bounded away from infeasible consumption and leisure. At an
interior solution its first-order condition is
\begin{equation}
\frac{wh}{c_{nr}}
=\nu(1-\ell-n\lambda)^{-\gamma}.
\label{eq:app_labor_foc}
\end{equation}
Conditional on $(n,r,x,e)$, labor supply is obtained from its first-order
condition whenever the implied choice is feasible; otherwise the household
chooses the feasible boundary.

\subsection{Ability and Human-Capital Transition}
\label{app:ability_transition}

Conditional on the parental family state $\kappa_p$, child ability follows
\[
\log\kappa'
=\rho_\kappa\log\kappa_p+\eta',
\qquad \eta'\sim N(0,\sigma_\kappa^2).
\]
The Tauchen grid contains $N_\kappa$ evenly spaced log points spanning
$\pm 2.5\sigma_\kappa/\sqrt{1-\rho_\kappa^2}$. For each parental state, an
interior child-ability point receives the normal probability between adjacent
grid midpoints. The first and last points also collect the lower and upper
tails, so every row of $\Pi_\kappa$ sums to one.
Given child ability $\kappa'$, productive human capital is
\[
h'(x,\kappa')=A_h\kappa'(\theta_0+x^{\alpha_1}).
\]
$A_h$ is normalized to one. The baseline component $\theta_0$ captures publicly
provided education and other human capital produced without private-education spending.

Conditional on state $s$, parity $n$, participation $r$, and ability $\kappa'$,
the destination indicator satisfies
\[
\Pr(D_E=1\mid s,n,r,\kappa')
=P_E\!\left(S(h'(x_{nr}(s),\kappa'),e_{nr}(s));v\right).
\]
The next adult state is
\begin{equation}
\widetilde h'
=h'(x_{nr}(s),\kappa')\exp(\delta_E D_E),
\qquad \kappa_p'=\kappa'.
\label{eq:app_next_income_mixture}
\end{equation}
Values of $\widetilde h'$ between earnings-grid points are assigned linearly to
the two adjacent states; values outside the grid are assigned to the nearest
endpoint. This preserves mass and makes the transition probabilities continuous
in policy-induced changes in $\widetilde h'$.

\subsection{Assignment and Capacity}
\label{app:assignment_capacity}

For each child ability realization, the score is
$S=(h')^{\theta_h}(1+e)$
and the assignment probability is
$P_E(S;v)
=\Lambda\!\left(\frac{S-v}{\sigma_s v}\right).$
Scaling the score gap by $v$ makes $\sigma_s$ a unit-free measure of assignment
noise. For each candidate cutoff $v$, I solve the household policies and their implied
stationary distribution, then compute the capacity residual
\begin{align}
\mathcal C(v)
=
\frac{
\sum_s\mu(s)\sum_{n=1}^3\sum_{r=0}^1\pi_{nr}(s)n
\sum_{\kappa'}\Pi_\kappa(\kappa'|\kappa_p)
P_E\!\left(S(h'(x_{nr}(s),\kappa'),e_{nr}(s));v\right)}
{\sum_s\mu(s)\sum_{n=1}^3\sum_{r=0}^1\pi_{nr}(s)n}
-q_E.
\label{eq:app_excess_destination}
\end{align}
The notation suppresses the dependence of $\mu$, $\pi_{nr}$, $x_{nr}$, and
$e_{nr}$ on $v$. The equilibrium cutoff solves $\mathcal C(v)=0$. Because each
child competes separately, the capacity condition weights each household by its
number of children: a family with $n$ children contributes $n$ applicants to
the pool. I locate a sign-changing bracket and solve the root with Brent's
method.

\subsection{Reproduction Matrix and Stationarity}
\label{app:reproduction_matrix}

Let $i=(h_i,\kappa_{p,i})$ and $j=(h_j,\kappa_{p,j})$ index current and
next-generation states. The entry $R_{ij}$ is the expected number of children
produced in state $i$ who enter state $j$. It averages over parity choice,
child ability, destination assignment, and interpolation on the $h$ grid:
\begin{align}
R_{ij}
=\sum_{n=1}^3\sum_{r=0}^1\pi_{nr}(i)n
\sum_{\kappa'}\Pi_\kappa(\kappa'|\kappa_{p,i})
\sum_{d\in\{0,1\}}
\Pr(d|i,n,r,\kappa')\,
\omega_j(\widetilde h'_{inr}(d,\kappa'),\kappa').
\label{eq:app_reproduction_matrix}
\end{align}
Here
$\widetilde h'_{inr}(d,\kappa')
=h'(x_{nr}(i),\kappa')\exp(\delta_Ed)$. The weight $\omega_j$ assigns
$\widetilde h'$ linearly to the two adjacent points on the $h$ grid while
setting $\kappa_{p,j}=\kappa'$. Destination assignment is Bernoulli:
\[
\begin{aligned}
\Pr(d=1|i,n,r,\kappa')
&=P_E\!\left(S(h'(x_{nr}(i),\kappa'),e_{nr}(i));v\right),\\
\Pr(d=0|i,n,r,\kappa')
&=1-\Pr(d=1|i,n,r,\kappa').
\end{aligned}
\]
Because the ability, assignment, and interpolation probabilities each sum to
one, $\sum_jR_{ij}=\sum_{n=1}^3\sum_{r=0}^1\pi_{nr}(i)n.$
Rows of $R$ therefore sum to expected births from the parental state, not to
one.
The vector $\mu$ gives the stable distribution of adults across states.
Applying $R$ to this distribution produces the next generation. In a stationary
composition, the next generation has the same state shares:
\[
\mu R=\varrho\mu,\qquad \sum_i\mu_i=1.
\]
Thus $\mu$ is the normalized nonnegative left eigenvector of $R$ associated
with the Perron root $\varrho$. Because $\mu$ sums to one, $\varrho$ equals
births per adult and is the factor by which cohort size changes from one
generation to the next.

\subsection{Equilibrium Algorithm}
\label{app:equilibrium_algorithm}

Given a trial cutoff $v$, I solve the conditional household problem in every
state and discrete branch, combine the solutions using the participation and
parity logits, and construct $R(v)$. Its normalized left Perron eigenvector gives
the stationary composition $\mu(v)$. Using these policies and this distribution,
I evaluate the capacity residual $\mathcal C(v)$ in
equation~\eqref{eq:app_excess_destination}. I first locate a sign-changing
bracket and then use Brent's method to solve $\mathcal C(v)=0$. Thus every trial
cutoff generates its own household policies and stationary distribution; the
benchmark does not use a damped fixed-point iteration over the cutoff and
distribution.
A numerical solution is accepted only if the capacity and normalized Perron
residuals are below their tolerances and consumption and leisure are positive
in every branch with positive probability. As consistency checks, I verify that
the nested-logit probabilities, each row of $\Pi_\kappa$, the interpolation
weights, and $\mu$ sum to one, and that each row of $R$ sums to expected births.

\subsection{Lottery Equilibrium}
\label{app:lottery_equilibrium}

Under the lottery, I replace $P_E(S;v)$ with $q_E$ in both expected utility and
the next-generation earnings transition. The cutoff and contest-participation
branch disappear, so households pay no entry fee and choose $e=0$, while
productive education remains valuable because it raises human capital. The
resulting policies define a new reproduction matrix, whose normalized left
Perron eigenvector gives the lottery's stationary distribution.

\subsection{Qualified Lottery}
\label{app:qualified_lottery}

For a qualified share $m\in[q_E,1]$, let
\[
Q(S;v_m)=\Lambda\!\left(\frac{S-v_m}{\sigma_s v_m}\right)
\]
denote the probability of clearing the qualification threshold. The cutoff
$v_m$ solves the capacity condition in
equation~\eqref{eq:app_excess_destination} with $P_E$ replaced by $Q$ and
$q_E$ replaced by $m$. Conditional on qualifying, each child receives the
destination with probability $q_E/m$, so the unconditional assignment
probability is
\begin{equation}
P_E^Q(S;v_m,m)=\frac{q_E}{m}Q(S;v_m).
\label{eq:app_qualified_assignment}
\end{equation}
I replace $P_E$ with $P_E^Q$ in the household problem and reproduction matrix
and solve the cutoff, policies, and stationary distribution jointly. At
$m=q_E$, the rule reproduces score-based assignment. At $m=1$, the threshold
is immaterial and the economy coincides with the full lottery above.

\subsection{Policy Equilibria}
\label{app:policy_equilibria}

For the separate tax and transfer experiments, I treat the lump-sum transfer
$T$ as an equilibrium object. Given a trial $T$, I solve household policies,
the capacity-clearing cutoff under score-based assignment, and the Perron
stationary distribution. I then vary $T$ until the government budget in
equation~\eqref{eq:government_budget} balances. A positive $T$ rebates net
revenue uniformly; a negative $T$ is a uniform lump-sum assessment.
The revenue-funded package instead fixes the broad tax at 20 percent and sets
$T=0$. Given a trial per-child transfer $b$, I re-solve the same equilibrium and
vary $b$ until aggregate transfer spending equals education-tax revenue. Thus
all proceeds from the tax are paid to families in proportion to their number of
children.

\section{Calibration Moments and Computation}
\label{app:calibration_computation}

This appendix maps the data moments into the model, states the calibration
criterion, and reports the numerical checks. Only benchmark moments determine
the parameter vector. The lottery and all other policies are solved afterward.

\subsection{Model Analogues of the Household Moments}
\label{app:model_moments}

Let $\mu_i$ denote stationary mass on adult state $i$, and let
$\pi_{nr}(i)$ be the joint probability of parity $n$ and participation state
$r\in\{0,1\}$. Labor income in that cell is
$y_{inr}=wh_i\ell_{nr}(i)$. Aggregate parity shares are
\begin{equation}
m_{p_n}=\sum_i\mu_i\sum_{r=0}^1\pi_{nr}(i),
\qquad n=0,1,2,3.
\label{eq:app_parity_moments}
\end{equation}
The model targets $m_{p_1},m_{p_2},m_{p_3}$; childlessness follows from adding-up.
Average completed fertility is
\begin{equation}
\overline n=\sum_i\mu_i\sum_{n=0}^3\sum_{r=0}^1n\pi_{nr}(i).
\label{eq:app_average_fertility}
\end{equation}

The income Gini is computed from the stationary distribution of $y_{inr}$ using
the same sorted cumulative-share formula as in KLIPS. Average labor supply is
\begin{equation}
m_\ell=\sum_i\mu_i\sum_n\sum_r\pi_{nr}(i)\ell_{nr}(i).
\label{eq:app_hours_moment}
\end{equation}
The fertility--income elasticity is the slope of log mean fertility on log mean
income across five equal-mass income quintiles, matching the data construction.
Bottom-quintile childlessness uses the same partition.

Let
$d_{inr}=x_{nr}(i)+r\{F+C(e_{nr}(i))\}$ denote private-education spending per
child. Because the model has no child-age dimension, the analogue of the
lifecycle spending target is the child-weighted mean of per-child spending as
a share of household income:
\begin{equation}
m_x
=\frac{
\sum_i\mu_i\sum_{n=1}^3\sum_{r=0}^1
\pi_{nr}(i)n\,d_{inr}/y_{inr}}
{\sum_i\mu_i\sum_{n=1}^3\sum_{r=0}^1\pi_{nr}(i)n}.
\label{eq:app_education_ratio}
\end{equation}
Both inputs enter because households finance productive education and positional
effort from the same budget. The 6.4 percent target fixes their combined resource
cost. The reported income elasticity of private-education spending is the slope
of log mean spending per child on log mean income across permanent-income
quintiles of parent households in the stationary distribution. It is reported
as an untargeted check because it draws on the same spending--income variation
as the targeted burden gradient.

Intergenerational income persistence is computed from the reproduction
matrix. Let $Y_p$ be log parental income and $Y_c$ be log child adult income. Then
\begin{equation}
m_{\mathrm{IGE}}
=\frac{\Cov(Y_p,Y_c)}{\Var(Y_p)}.
\label{eq:app_ige}
\end{equation}
The covariance is computed across all parent--child pairs generated in
equilibrium. Each pair is weighted by the stationary mass of the parental state,
the joint probability of fertility and contest participation, the number of
children, the probability of the child's ability realization, and the
destination-assignment probability.

\subsection{Joint Parameter Selection}
\label{app:joint_parameter_selection}

The vector of internally calibrated parameters has sixteen elements,
\begin{equation}
\Theta
=
(\phi_1,\phi_2,\phi_3,\beta,\sigma,\nu,
\sigma_\kappa,\rho_\kappa,B_E,c_e,\sigma_s,
\theta_h,\alpha_1,\theta_0,\psi,F).
\label{eq:app_parameter_vector}
\end{equation}
The parameters fall into three blocks. The preference block contains the parity
preference levels $(\phi_1,\phi_2,\phi_3)$, the weight on child outcomes $\beta$,
fertility-choice dispersion $\sigma$, and the leisure weight $\nu$. The ability
block $(\sigma_\kappa,\rho_\kappa)$ governs dispersion and intergenerational
persistence. The remaining block covers destination value, the contest, and
human-capital production: $B_E$ is destination value beyond current wages,
$c_e$ scales positional costs, $\sigma_s$ controls both participation dispersion
and assignment noise, $\theta_h$ loads human capital into the admission score,
$(\alpha_1,\theta_0)$ govern human-capital production, $\psi$ prices child study
time, and $F$ is the contest entry fee. All sixteen parameters are selected
jointly. Since the model moments depend jointly on $\Theta$, there is no
one-to-one link between individual parameters and moments. Some moments are
nevertheless more informative for particular parameters, as described below,
even though there is no formal identification procedure. The fixed child-income
curvature $\eta_c=2$ lies outside $\Theta$ and
matches the leisure curvature $\gamma=2$. The three parity levels are optimized
in levels, all strictly positive parameters are optimized in logs, and
$\rho_\kappa$ is bounded between 0.01 and 0.99.

The calibration uses ten point targets and two interval restrictions. The point
targets are the shares with one, two, and three or more children, the income
Gini, the work-time share, the fertility--income elasticity, bottom-quintile
childlessness, the private-education spending-to-income ratio,
intergenerational income persistence, and the private-education burden
gradient. Each point-target residual is the model's proportional deviation from
its empirical target:
\begin{equation}
r_j(\Theta)
=\frac{m_j(\Theta)-\widehat m_j}{\widehat m_j},
\qquad j=1,\ldots,10 .
\label{eq:app_resid_relative}
\end{equation}
The parity shares are especially informative about the parity preference levels,
the weight on child outcomes, and fertility-choice dispersion. Work hours are
most informative about the leisure weight, while the income Gini and
intergenerational persistence discipline ability dispersion and persistence.
The fertility--income elasticity and bottom-quintile childlessness further
restrict preference heterogeneity. The spending level and burden gradient,
together with the KEEP restrictions, jointly discipline the human-capital and
contest parameters. These links are informative rather than exclusive because
every parameter affects equilibrium choices and the stationary distribution.

The two KEEP point estimates are close to zero. Because their 95 percent
confidence intervals are wide, they enter as interval restrictions rather than
exact targets. The estimated wage return to private-education spending is
$\widehat b_w=0.0017$ with a standard error of $0.0125$, and the estimated logit
slope of preferred-destination entry with respect to log private-education
spending is $\widehat b_g=0.027$ with a standard error of $0.059$. The
calibration therefore assigns zero loss to model values inside the corresponding
95 percent confidence intervals.

Both model analogues depend on the productive-return term
$\alpha_1s$, where
$s=x^{\alpha_1}/(\theta_0+x^{\alpha_1})$. For the wage-return analogue,
productive education raises earnings directly through human capital and
indirectly by improving destination assignment. For this aggregate analogue, I
evaluate the probability term at the admission share $q_E$ and the score ratio
$S/v$ at its cutoff value of one. Because assignment adds $\delta_E$ to log
income, the indirect component is
$\delta_E q_E(1-q_E)(\theta_h/\sigma_s)\alpha_1s$, where $q_E(1-q_E)$ converts
the change in the logit index into a change in assignment probability. I average
the direct and indirect components over the stationary parent distribution. The
indirect component is $0.0006$ at the benchmark.

The destination-entry analogue instead evaluates each household at its actual
score and is proportional to
$(S/v)(\theta_h/\sigma_s)\alpha_1s$. At the benchmark,
$\theta_h/\sigma_s=1.01$ and mean $S/v$ is $0.61$. The wage restriction
primarily disciplines the productive return to education, while the entry
restriction additionally disciplines the loading of human capital in the
assignment score.

KEEP estimates the wage return to total private-education spending among
employed 29-year-olds, whereas the model analogue is the stationary average wage
response to productive education $x$. To place them on the same scale, I map a
1 percent increase in total private-education spending to a 1 percent increase
in $x$. Because some observed spending is positional and does not build human
capital, this mapping gives an upper bound on the human-capital component of the
return.

Index the two KEEP restrictions by $j\in\{11,12\}$. Let $b_j(\Theta)$ denote the
corresponding model analogue and let $\widehat b_j$ and
$\operatorname{se}_j$ denote the empirical estimate and its standard error.
Define $t_j(\Theta)=[b_j(\Theta)-\widehat b_j]/\operatorname{se}_j$ and
\begin{equation}
r_j(\Theta)
=
\begin{cases}
0, & |t_j(\Theta)|\le 1.96,\\[2pt]
\bigl(|t_j(\Theta)|-1.96\bigr)\,\operatorname{sign}\{t_j(\Theta)\}, & \text{otherwise,}
\end{cases}
\qquad j=11,12.
\label{eq:app_resid_interval}
\end{equation}
Each residual is zero inside its 95 percent confidence interval and equals the
standardized distance from the nearest endpoint outside. The minimum-distance
criterion is $r(\Theta)'r(\Theta)$, combining the ten relative point-target
residuals with these two interval residuals.

\subsection{Measurement of the Positional Share}
\label{app:positional_share_bounds}

The positional share is the fraction of education spending that raises admission
chances without raising human capital. In the 2017 and 2018 Private Education 
Expenditures Surveys, admission preparation, anxiety, and advance study
account for 0.299 of general-subject spending. This is the narrow lower bound on 
positional spending.
Another 0.66 is reported as supplementing or deepening school lessons and may
build human capital, protect rank, or do both. At the upper end, the only spending clearly
unrelated to the race is supervision or sociality, which accounts for 0.032 and
0.033 of spending. The resulting empirical range is [0.30, 0.97].
The calibration does not target either endpoint. The calibrated model produces a benchmark share
of 0.773.

\subsection{Numerical Procedure and Checks}
\label{app:optimization_sequence}
\label{app:grid_checks}

I minimize the calibration criterion with a box-constrained, derivative-free
search that solves a stationary equilibrium at each candidate parameter vector.
The parameter vector with the lowest criterion value obtained in the numerical
search is the benchmark. The benchmark and all counterfactuals use the same grid and
convergence criteria. The final grid has 48 earnings states and 17 ability
states. The capacity cutoff and stationary distribution are solved to
a tolerance of $5\times10^{-4}$. The equilibrium capacity residual is
$4.4\times10^{-6}$, and neither productive education nor positional effort
reaches its numerical upper bound in any positive-mass state. All calibrated
parameters lie inside their search bounds; $\alpha_1$ is closest, at 15.1
percent of the transformed interval above its lower bound.

\subsection{Parameter Sensitivity}
\label{app:handset_sensitivity}

The lottery effect is insensitive to the positional-cost curvature and the
numerical bounds. Varying the positional-cost curvature from 0.75 to 3.0 changes
the fertility effect only from 0.243 to 0.238, while doubling the upper bound on
education spending leaves it at 0.240. Doubling the upper support of the earnings
grid eliminates the 2.86 percent mass at its upper endpoint and also leaves the
effect at 0.240. The calibration fit is nearly flat as the child-income curvature
varies from 0.9 to 4.4 because $B_E$ adjusts in the opposite direction, so I set
$\eta_c=2$, matching the leisure curvature. The child-time curvature matters
more: doubling it lowers the fertility effect from 0.240 to 0.169, which remains
more than five times the effect of the five-percent pronatal transfer.

\begingroup
\setlength{\intextsep}{6pt plus 1pt minus 1pt}

\section{Capacity and Assignment}
\label{app:capacity_results}

Table~\ref{tab:app_capacity} compares capacity changes under score-based
assignment with the same changes combined with random assignment. Within each
$\eta$ block, the first row reports capacity alone and the second adds the
lottery. For every reported $q_E$ and $\eta$, random assignment raises
fertility relative to score-based assignment at the same capacity. This
additional gain becomes small as $q_E$ approaches one because near-universal
access leaves little incentive to compete for rank under either assignment rule.

\begin{table}[H]
\centering\footnotesize
\caption{Capacity alone and capacity with random assignment}
\label{tab:app_capacity}
\begin{threeparttable}
\begin{tabular}{lrrrrrrr}
\toprule
$\eta$ & \multicolumn{7}{c}{$q_E$} \\
\cmidrule(lr){2-8}
 & 0.1 & 0.184 & 0.25 & 0.35 & 0.5 & 0.7 & 0.95 \\
\midrule
0.00, capacity alone & -0.048 & -0.000 & +0.036 & +0.105 & +0.269 & +0.648 & +1.093 \\
\quad + lottery & +0.037 & +0.239 & +0.380 & +0.571 & +0.808 & +1.009 & +1.112 \\
0.25, capacity alone & -0.028 & -0.000 & +0.025 & +0.077 & +0.201 & +0.483 & +0.920 \\
\quad + lottery & +0.078 & +0.239 & +0.343 & +0.478 & +0.645 & +0.817 & +0.962 \\
0.50, capacity alone & -0.005 & -0.000 & +0.014 & +0.050 & +0.142 & +0.342 & +0.667 \\
\quad + lottery & +0.124 & +0.239 & +0.308 & +0.393 & +0.495 & +0.603 & +0.710 \\
0.75, capacity alone & +0.019 & -0.000 & +0.004 & +0.025 & +0.089 & +0.224 & +0.448 \\
\quad + lottery & +0.177 & +0.240 & +0.275 & +0.317 & +0.366 & +0.418 & +0.469 \\
1.00, capacity alone & +0.046 & +0.000 & -0.007 & +0.003 & +0.043 & +0.129 & +0.275 \\
\quad + lottery & +0.235 & +0.240 & +0.244 & +0.249 & +0.258 & +0.269 & +0.283 \\
\bottomrule
\end{tabular}
\end{threeparttable}
\begin{fullwidthtablenotes}
Entries are changes in completed fertility from the benchmark. The $q_E=0.10$
column is a capacity contraction.
\end{fullwidthtablenotes}\end{table}

\endgroup

\section{Welfare: Criterion, Metrics, and Additional Results}
\label{app:welfare}

This appendix develops the welfare analysis in Section~\ref{sec:welfare}.
It defines the criterion, decomposes the baseline consumption-based measure,
reports the additional child-time effect, derives the transition law, and examines
sensitivity to destination value and capacity. Unless stated otherwise, the
admission share is held fixed as cohort size changes.

\subsection{Welfare Criterion}

The criterion asks whether an adult would prefer the reform before learning
income and parental background. For state $s=(h,\kappa_p)$, let
$\mathcal R_0=\{0\}$ and $\mathcal R_n=\{0,1\}$ for $n\geq1$. The two nested
inclusive values are
\begin{align}
\overline U_n(s)
&=\sigma_s\log\!\sum_{r\in\mathcal R_n}
  \exp\!\left(\frac{V_{nr}(s)}{\sigma_s}\right),\notag\\
W
&=\int \sigma\log\!\sum_{n=0}^{3}
  \exp\!\left(\frac{\overline U_n(s)}{\sigma}\right)\dd\mu(s),
\label{eq:welfare_criterion}
\end{align}
where $V_{nr}$ is conditional utility at the optimal continuous choices,
$\sigma_s$ is the participation-shock scale, and $\sigma$ is the
fertility-shock scale. The integral averages over the stationary adult
distribution. Because log consumption enters with coefficient one, a uniform
proportional consumption compensation $\lambda$ shifts every conditional value
by $\log(1+\lambda)$. The exact consumption equivalent is therefore
$\mathrm{CE}=\exp(\Delta W)-1$. All welfare measures are per adult within a calibrated cohort economy. Levels
are not compared across cohorts because their calibrated taste dispersions
differ.

\subsection{Baseline Welfare and the Child-Time Effect}

The baseline welfare measure is the partial-accounting consumption equivalent
that excludes the child-time term. I add back each regime's expected child-time
cost at that regime's equilibrium choices before comparing welfare because the
two education inputs have no estimated conversion into market hours.
Allocations and the stationary distribution still respond to the full utility
function, so this is not the welfare of a model re-solved without the
child-time term.

At a fixed admission share, the lottery raises baseline welfare by
\vCEexTime\ percent of consumption. Table~\ref{tab:welfare_decomp} decomposes
this gain. Higher fertility increases the expected nonwage value of children
reaching preferred destinations because the admission probability remains
$q_E$ per child. The change in direct family-size utility and the logit choice
surplus offset part of that gain. The residual combines consumption, leisure,
and the wage component of child outcomes; it includes the resources released
from entry fees and positional outlays but is not a pure budget term.

\begin{table}[htbp]
\centering\footnotesize
\caption{Baseline welfare and the additional child-time gain}
\label{tab:welfare_decomp}
\begin{threeparttable}
\begin{tabular}{lr}
\toprule
Component & Lottery $-$ contest \\
\midrule
\multicolumn{2}{l}{\emph{Baseline welfare measure}} \\
Nonwage destination-value flow & $+0.427$ \\
Change in direct family-size utility & $-0.255$ \\
Logit choice-surplus change & $-0.226$ \\
Other household-utility terms (residual) & $+0.205$ \\
\cmidrule(lr){1-2}
Baseline $\Delta W$ & $+0.151$ \\
Consumption equivalent (\%) & $+16.3$ \\
\addlinespace
\multicolumn{2}{l}{\emph{Additional child-time effect}} \\
Child-time cost released & $+1.664$ \\
Full-model $\Delta W$ & $+1.815$ \\
\bottomrule
\end{tabular}
\end{threeparttable}
\begin{fullwidthtablenotes}
The first four rows are changes in model utility and add to baseline
$\Delta W$; the consumption equivalent is $\exp(\Delta W)-1$.
The nonwage destination-value row rises with the expected number of
children reaching preferred destinations at a fixed per-child admission
probability. The residual includes the resource release from lower entry
fees and positional outlays. The child-time row is then added to obtain
full-model $\Delta W$. It is reported separately because the two education
inputs have no estimated conversion into market hours.
\end{fullwidthtablenotes}\end{table}

The baseline measure includes the two logit inclusive values. Weighting
conditional utility by the choice probabilities instead, without the logit
surplus, raises the consumption equivalent from \vCEexTime\ to
\vCEexTimeNoShock\ percent but leaves its sign unchanged.

The model also assigns the child-time cost
$n\psi(x+e)^{1+\varphi}/(1+\varphi)$ to productive education and positional
preparation. The lottery releases \vChildTimeRelease\ additional utility units
by nearly eliminating this cost, bringing the full-model increase to
\vDWfull\ utility units. The child-time release accounts for
\vChildTimeShare\ percent of that increase. In the model, the largest
full-utility benefit therefore comes from children no longer spending time on
positional preparation. This is a modeled reduction in children's time burden,
not a direct estimate of their subjective well-being.

Although excluded from the baseline accounting measure, the child-time term
continues to affect equilibrium choices. Holding the calibrated parameter
vector fixed, doubling its curvature changes the baseline lottery gain from
\vCEexTime\ to \vCEctcHi\ percent, with no change in sign. At half curvature,
the equilibrium does not converge.

\subsection{Generation Path}

The transition starts from the benchmark stationary adult distribution. These
adults were themselves raised under the contest but make fertility and
education-investment choices for their children under the lottery. Because
lottery policies do not depend on an equilibrium cutoff or fiscal transfer, the
adult distribution evolves as $\mu_{t+1}\propto\mu_tR$, where $R$ is the lottery
reproduction matrix defined in Appendix~\ref{app:reproduction_matrix}.

\subsection{Sensitivity to Destination Value and Capacity}

\emph{Destination value lost under random assignment.} The main lottery
experiment assumes that the measured wage premium survives random assignment.
If the entire premium instead reflects score-based matching and disappears
under the lottery, the baseline welfare gain falls from \vCEexTime\ to
\vCEphiOne\ percent, approximately zero. In a broader sensitivity check, I
reduce the wage and nonwage components of destination value proportionally
under the lottery. The gain disappears when both components fall by
\vPhiTotal\ percent. This low threshold reflects the fact that the baseline
gain is the net of the offsetting welfare changes reported in
Table~\ref{tab:welfare_decomp}.

\emph{Fixed number of positions.} The main lottery holds the admission share
$q_E$ fixed as fertility changes, so the number of preferred positions grows
with cohort size. If the number of positions is held fixed instead, the
steady-state admission probability falls from 18.45 to 16.65 percent.
Fertility still rises by 0.199 children, but the baseline welfare change is
\vCEabsSlots\ percent. The positive baseline welfare result therefore depends
on maintaining the admission share, whereas the fertility response remains
positive under the fixed-number convention.

\end{document}